\pdfoutput=1
\documentclass[acmsmall,screen]{acmart}

\setcopyright{rightsretained}
\acmPrice{}
\acmDOI{10.1145/3563297}
\acmYear{2022}
\copyrightyear{2022}
\acmSubmissionID{oopslab22main-p217-p}
\acmJournal{PACMPL}
\acmVolume{6}
\acmNumber{OOPSLA2}
\acmArticle{134}
\acmMonth{10}

\ccsdesc[300]{Theory of computation~Data structures design and analysis}
\ccsdesc[500]{Computer systems organization~Quantum computing}
\ccsdesc[100]{Software and its engineering~Language features}
\ccsdesc[100]{Software and its engineering~Formal language definitions}
\ccsdesc[100]{Software and its engineering~Memory management}

\keywords{quantum programming, data structures, quantum random-access memory, reversible programming, history independence}

\usepackage{booktabs}
\usepackage{caption}
\usepackage[defaultlines=2,all]{nowidow}
\usepackage{textgreek}

\usepackage[utf8]{inputenc} 
\usepackage[T1]{fontenc}    
\usepackage{amsfonts}       
\usepackage{nicefrac}       
\usepackage{microtype}      
\usepackage{amsthm}
\usepackage{xspace}
\usepackage[unicode]{hyperref}
\usepackage[capitalise]{cleveref}
\usepackage{enumitem}
\usepackage{multirow}
\usepackage{wrapfig}
\usepackage{subcaption}
\usepackage{threeparttable}

\usepackage{tikz}
\usetikzlibrary{positioning}
\usetikzlibrary{calc}
\usetikzlibrary{quantikz}

\usepackage{listings}
\lstdefinelanguage{tower}
{morekeywords={
     fun,let,bool,unit,uint,ptr,type,if,else,return,with,do,not,test,null,alloc,default,static,land,lor,lnot,lsl,lsr
   },
 sensitive=true,
 morecomment=[n]{/*}{*/},
 morestring=[b]",
 escapechar=\%,
 columns=fullflexible,
 keepspaces=true,
 basicstyle=\ttfamily,
 mathescape=true,
}
\lstdefinestyle{color}{
    keywordstyle=\color{blue!50!black}\bfseries\ttfamily,
}
\lstdefinestyle{nonum}{
    numbers=none,
}
\lstdefinestyle{tiny}{
    basicstyle=\tiny\ttfamily,
}
\makeatletter
\lst@AddToHook{EveryPar}{%
    \edef\@temp{\noexpand\label{\lst@label-\arabic{lstnumber}}}%
    \@temp
}
\makeatother

\usepackage{mathpartir}
\usepackage{xargs}
\usepackage{pifont}
\usepackage{balance}
\usepackage{xcolor}

\usepackage{mathtools}
\usepackage{stmaryrd}
\usepackage{bbm}
\usepackage{cancel}
\usepackage{placeins}
\usepackage[normalem]{ulem}

\usepackage{subfiles}

\makeatletter
\newcommand*{\subfilesbibliography}[1]{%
 \expandafter\ifx\csname ver@subfiles.cls\endcsname\relax
   \expandafter\@secondoftwo
 \else
   \expandafter\@firstoftwo
  \fi
  {\bibliography{#1}}
  {}%
}
\makeatother

\newcommand{\LangName}{Tower}
\newcommand{\CoreLang}{Core \LangName{}}
\newcommand{\AllocatorName}{Boson}
\newcommand{\LibraryName}{Ground}

\title{\LangName{}: Data Structures in Quantum Superposition}

\author{Charles Yuan}
\affiliation{
  \institution{MIT CSAIL}
  \streetaddress{32 Vassar St}
  \city{Cambridge}
  \state{MA}
  \postcode{02139}
  \country{USA}
}
\email{chenhuiy@csail.mit.edu}
\orcid{0000-0002-4918-4467}

\author{Michael Carbin}
\affiliation{
  \institution{MIT CSAIL}
  \streetaddress{32 Vassar St}
  \city{Cambridge}
  \state{MA}
  \postcode{02139}
  \country{USA}
}
\email{mcarbin@csail.mit.edu}
\orcid{0000-0002-6928-0456}

\DeclarePairedDelimiter\abs{\lvert}{\rvert}%

\newcommand{\ctx}{\ensuremath{\Gamma}}
\newcommand{\tctx}{\ensuremath{\Delta}}

\newcommand{\type}{\ensuremath{\tau}}
\newcommand{\tUnit}{\ensuremath{()}}

\newcommand{\tUInt}{\ensuremath{\texttt{uint}}}
\newcommand{\tBool}{\ensuremath{\texttt{bool}}}
\newcommand{\tPair}[2]{\ensuremath{({#1}, {#2})}}
\newcommand{\tPtr}[1]{\ensuremath{\texttt{ptr}({#1})}}
\newcommand{\tInd}[2]{\ensuremath{\nu {#1}. {#2}}}

\newcommand{\vValue}{\ensuremath{v}}
\newcommand{\vUnit}{\ensuremath{()}}
\newcommand{\vNum}[1]{\ensuremath{\overline{#1}}}
\newcommand{\vTrue}{\ensuremath{\texttt{true}}}
\newcommand{\vFalse}{\ensuremath{\texttt{false}}}
\newcommand{\vNull}[1]{\ensuremath{\texttt{null}_{#1}}}
\newcommand{\vPtr}[2]{\ensuremath{\texttt{ptr}_{#1}[{#2}]}}

\newcommand{\eExp}{\ensuremath{e}}
\newcommand{\ePair}[2]{\ensuremath{({#1},{#2})}}
\newcommand{\eUnop}[2]{\ensuremath{{#1}\ {#2}}}
\newcommand{\eBinop}[3]{\ensuremath{{#1}\ {#2}\ {#3}}}
\newcommand{\eProj}[2]{\ensuremath{\pi_{#1}({#2})}}
\newcommand{\eAlloc}[1]{\ensuremath{\texttt{alloc}({#1})}}
\newcommand{\eCall}[3]{\ensuremath{{#1}^{#2}({#3})}}

\newcommand{\dDecl}{\ensuremath{d}}
\newcommand{\dFun}[5]{\ensuremath{\texttt{fun}\ {#1}^{#2}\ ({#3}) \to {#4}\ \texttt{is}\ {#5}}}

\newcommand{\oNot}{\ensuremath{\texttt{not}}}
\newcommand{\oTest}{\ensuremath{\texttt{test}}}
\newcommand{\oAdd}{\ensuremath{\texttt{+}}}
\newcommand{\oMul}{\ensuremath{\texttt{*}}}
\newcommand{\oSub}{\ensuremath{\texttt{-}}}
\newcommand{\oAnd}{\ensuremath{\texttt{\&\&}}}
\newcommand{\oOr}{\ensuremath{\texttt{||}}}

\newcommand{\sStmt}{\ensuremath{s}}
\newcommand{\sSkip}{\ensuremath{\texttt{skip}}}
\newcommand{\sSeq}[2]{\ensuremath{{#1};\,{#2}}}
\newcommand{\sBind}[2]{\ensuremath{{#1} \leftarrow {#2}}}
\newcommand{\sUnbind}[2]{\ensuremath{{#1} \rightarrow {#2}}}
\newcommand{\sSwap}[2]{\ensuremath{{#1} \Leftrightarrow {#2}}}
\newcommand{\sMemSwap}[2]{\ensuremath{*{#1} \Leftrightarrow {#2}}}
\newcommand{\sIf}[2]{\ensuremath{\texttt{if}\ {#1}\ \texttt{then}\ {#2}\ \texttt{end}}}
\newcommand{\sReturn}[1]{\ensuremath{\texttt{return}\ {#1}}}

\newcommand{\sWith}[2]{\ensuremath{\texttt{with}\ {#1}\ \texttt{do}\ {#2}\ \texttt{end}}}

\newcommand{\bBound}{\ensuremath{\beta}}
\newcommand{\bVar}{\ensuremath{b}}

\newcommand{\pProgram}{\ensuremath{p}}

\newcommand{\typeok}[2]{\ensuremath{{#1} \vdash {#2}\ \textsf{ok}}}
\newcommand{\typeequiv}[2]{\ensuremath{{#1} = {#2}}}
\newcommand{\hastype}[3]{\ensuremath{{#1} \vdash {#2} : {#3}}}
\newcommand{\stmtok}[3]{\ensuremath{{#1} \vdash {#2} \dashv {#3}}}
\newcommand{\modified}[1]{\ensuremath{\textsf{mod}({#1})}}
\newcommand{\exposed}[1]{\ensuremath{\textsf{ex}({#1})}}

\newcommand{\fctx}{\ensuremath{\Phi}}
\newcommand{\hastypef}[5]{\ensuremath{{#1}; {#2} \vdash_{#3} {#4} : {#5}}}
\newcommand{\stmtokf}[5]{\ensuremath{{#1}; {#2} \vdash_{#3} {#4} \dashv {#5}}}

\newcommand{\get}[2]{\ensuremath{{#1}[{#2}]}}
\newcommand{\subst}[3]{\ensuremath{{#1}[{#2} \leftarrow {#3}]}}

\newcommand{\default}[1]{\ensuremath{\textsf{z}({#1})}}

\newcommand{\mstate}[1]{\ensuremath{\langle{#1}\rangle}}
\newcommand{\reg}{\ensuremath{R}}
\newcommand{\mem}{\ensuremath{M}}
\newcommand{\estep}[5]{\ensuremath{\mstate{{#1}, \subst{#2}{#3}{#4}} \Downarrow_{#3} \mstate{{#1}, \subst{#2}{#3}{#5}}}}
\newcommand{\estepline}[5]{\mstate{{#1}, \subst{#2}{#3}{#4}} \Downarrow_{#3} \\\\ \mstate{{#1}, \subst{#2}{#3}{#5}}}
\newcommand{\sstep}[6]{\ensuremath{\mstate{{#1}, {#2}, {#3}} \mapsto \mstate{{#4}, {#5}, {#6}}}}
\newcommand{\sstepline}[6]{\mstate{{#1}, {#2}, {#3}} \mapsto \\\\ \mstate{{#4}, {#5}, {#6}}}
\newcommand{\sstepstar}[6]{\ensuremath{\mstate{{#1}, {#2}, {#3}} \mapsto^* \mstate{{#4}, {#5}, {#6}}}}
\newcommand{\restep}[5]{\ensuremath{\mstate{{#1}, \subst{#2}{#3}{#5}} \Uparrow_{#3} \mstate{{#1}, \subst{#2}{#3}{#4}}}}
\newcommand{\rsstep}[6]{\ensuremath{\mstate{{#1}, {#5}, {#6}} \mapsto_R \mstate{{#4}, {#2}, {#3}}}}
\newcommand{\rsstepstar}[6]{\ensuremath{\mstate{{#1}, {#5}, {#6}} \mapsto_R^* \mstate{{#4}, {#2}, {#3}}}}
\newcommand{\sstepstuck}[3]{\ensuremath{\mstate{{#1}, {#2}, {#3}}\ \mathsf{err}}}
\newcommand{\rsstepstuck}[3]{\ensuremath{\mstate{{#1}, {#2}, {#3}}\ \mathsf{err}_R}}
\newcommand{\rsstepline}[6]{\mstate{{#1}, {#5}, {#6}} \mapsto_R \\\\ \mstate{{#4}, {#2}, {#3}}}

\newcommand{\regtypes}[1]{\ensuremath{\ctx_{#1}}}

\newcommand{\denote}[1]{\ensuremath{\left\llbracket{#1}\right\rrbracket}}

\newcommand{\reverse}[1]{\mathcal{I}[{#1}]}
\newcommand{\circuit}[1]{\mathcal{C}\denote{#1}}

\makeatletter
\newcommand{\oset}[3][0ex]{%
  \mathrel{\mathop{#3}\limits^{
    \vbox to#1{\kern-0.75\ex@
    \hbox{$\scriptstyle#2$}\vss}}}}
\makeatother

\newcommand{\qstate}{\ensuremath{\psi}}

\newcommand{\pstate}[1]{\smash{\colorbox{gray!20}{\ensuremath{#1}\vphantom{S}}}\makebox(0,9){}}
\newcommand{\goesto}{\ensuremath{\hookrightarrow}}

\newcommand{\Comment}[1]{\ignorespaces}

\begin{document}
\def\biblio{}

\begin{abstract}
Emerging quantum algorithms for problems such as element distinctness, subset sum, and closest pair demonstrate computational advantages by relying on \emph{abstract data structures}. Practically realizing such an algorithm as a program for a quantum computer requires an efficient implementation of the data structure whose operations correspond to \emph{unitary operators} that manipulate quantum \emph{superpositions} of data.

To correctly operate in superposition, an implementation must satisfy three properties --- \emph{reversibility}, \emph{history independence}, and \emph{bounded-time execution}. Standard implementations, such as the representation of an abstract set as a hash table, fail these properties, calling for tools to develop specialized implementations.

In this work, we present \CoreLang{}, the first language for quantum programming with random-access memory. \CoreLang{} enables the developer to implement data structures as pointer-based, linked data. It features a reversible semantics enabling every valid program to be translated to a unitary quantum circuit.

We present \AllocatorName{}, the first memory allocator that supports reversible, history-independent, and constant-time dynamic memory allocation in quantum superposition.
We also present \LangName{}, a language for quantum programming with recursively defined data structures. \LangName{} features a type system that bounds all recursion using classical parameters as is necessary for a program to execute on a quantum computer.

Using \LangName{}, we implement \LibraryName{}, the first quantum library of data structures, including lists, stacks, queues, strings, and sets. We provide the first executable implementation of sets that satisfies all three mandated properties of reversibility, history independence, and bounded-time execution.

\end{abstract}

\maketitle

\section{Introduction} \label{sec:intro}

Algorithms for quantum computers promise computational advantages in areas ranging from integer factorization~\citep{shor1997} and search~\citep{grover} to cryptography~\citep{bennett}, linear algebra~\citep{Harrow_2009}, physical simulation~\citep{simulation-physics,simulation-chem}, and combinatorial optimization~\citep{farhi2014}. Enabling practical realization of these algorithms on quantum hardware is a major aim of \emph{quantum programming languages}, in which a developer may write programs that execute on a quantum computer.

The principles of quantum mechanics specify how a quantum computer manipulates \emph{quantum states} consisting of \emph{qubits}, the quantum analogs of classical bits. A quantum state exists in a \emph{superposition}, a weighted sum over classical states. The computer may evolve the state using \emph{unitary operators}, or it may \emph{measure} the state, which causes a superposition to assume a classical state with probability derived mathematically from the weight ascribed to that state in the sum.

\subsection{Data Structures and Quantum Algorithms}

Existing quantum programming languages~\citep{qsharp, quipper, qwire, liquid, qml, quantum-lambda-calc, silq} build abstractions for individual qubits and primitive data types such as integers.
By contrast, emerging quantum algorithms incorporate abstract data structures for sake of implementation and efficiency. The concrete representation of such a data structure classically is a linked structure in random-access memory, defined recursively and constructed by dynamic allocation.

For example, the set data structure can maintain a duplicate-free collection of items, check an element for membership in the collection, and add and remove elements from the collection.
Examples of quantum algorithms that depend critically on sets include \citet{ambainis2003}, who presents an $O(n^{2/3})$ quantum algorithm for element distinctness that improves over the $\Omega(n)$ classical bound; \citet{bernstein2013}, who present a quantum algorithm for subset sum that beats the classical complexity of $O(2^{n/4})$; and \citet{aaronson2019}, who present an $\tilde{O}(n^{2/3})$ quantum algorithm for closest pair, improving over $O(n\log n)$ classically.

Each of these algorithms relies on an implementation of sets whose operations execute in $O(\log^c n)$ time. Classically, random-access memory is used to attain this efficient time complexity, and these algorithms similarly invoke its quantum analogue.
However, to date, no quantum programming language provides abstractions for constructing and manipulating sets and other abstract data types in random-access memory, a hurdle to the realization of these algorithms.

\subsection{Three Requirements for Data Structures}

Importantly, not all classical data structure implementations are compatible with quantum programming. \citet{ambainis2003} identifies a data structure appropriate for use in a quantum algorithm as one that satisfies three properties: \emph{reversibility}, \emph{history independence} and \emph{bounded-time execution}.

\paragraph{Reversible Operation}

The operations of the data structure must each be \emph{reversible}, meaning it can be executed forwards or backwards to produce or undo its effect on the program state.

In order to execute on a quantum computer, a program must first be translated into a sequence of primitive \emph{quantum gates}. Each gate corresponds to either a measurement or a unitary operator over the program state. Conventional operations over a data structure in superposition cannot safely invoke measurement, because measurement does not preserve the superposition of program data. Instead, it irreversibly collapses the data into a classical state, which would cause algorithms such as~\citet{ambainis2003,bernstein2013,aaronson2019} to return incorrect results.

Thus, when a program executes an operation on data in superposition, the operation must correspond to a unitary operator, which has an inverse by definition. As a consequence, the operations of a data structure in superposition must be reversible.

\paragraph{History-Independent Representation}

The representation of the data structure must be \emph{history-independent}, meaning that each instance of the data structure is assigned a single unique physical representation in memory that is insensitive to the history of operations performed on it.

History independence is essential for the correctness or performance of algorithms such as~\citet{ambainis2003,aaronson2019,bernstein2013}. The reason is that they rely on quantum \emph{interference}, the phenomenon in which weights of identical states in a superposition combine or cancel, altering the outcomes of measurements and the output of the algorithm.

For example, the algorithm of~\citet{ambainis2003} generates a superposition over sets of integers. Suppose that among these are two instances of the set $\{1,2\}$ whose weights in the superposition have equal complex magnitude but opposite sign. Then, the correctness of the algorithm relies on quantum interference to guarantee that the probability of observing the set $\{1,2\}$ cancels to zero.

However, suppose that sets are concretely represented using linked lists, such that the first instance is represented as $[1,2]$ and the second as $[2,1]$. Because the memory representations of these two instances are physically distinct, their weights do not cancel, leading the algorithm to produce incorrect outputs.
Thus, for the algorithm to execute correctly, each semantic instance of an abstract data structure must be assigned a unique underlying physical representation.

\paragraph{Bounded-Time Execution}

A data structure operation must execute in \emph{bounded time}, meaning its execution time is bounded by classical parameters not dependent on quantum data.

A recursive definition of a data structure typically corresponds to recursive implementations of its operations. Recursion enables an operation to use variable time complexity to traverse and manipulate a variable-sized data structure.
However, a quantum computer does not natively support recursive control flow and must first classically convert the program to a loop-free form. Hence, all recursion in the program must be bounded by classically determined parameters.

\subsection{\LangName{}: Data Structures in Quantum Superposition}

\paragraph{\CoreLang}

We present \CoreLang{}, the first language that enables quantum programming with data in random-access memory. It enables the construction and manipulation of pointer-based, linked data, by means of operations for integers, pointers, and control flow.

\paragraph{Reversibility}

In \CoreLang{}, every program that can execute forwards can also execute backwards.
A developer may compute an expression by assigning it to a variable, or \emph{uncompute} an expression, which reverses the computation, by \emph{un-assigning} it from a variable. Other features enable the developer to swap two variables, or to swap a variable with the content of memory at an address.

We present a fully reversible operational semantics of \CoreLang{}. Each operator has a reverse within the syntax, meaning a developer may write a program and derive another program with reversed semantics using the syntactic transformation of \emph{program inversion}.

We provide a translation from a valid \CoreLang{} program to a unitary operator that can execute on a quantum computer. Specifically, the operator corresponding to a \CoreLang{} program acts uniformly across each classical program state in a superposition. The operator may then be embedded within a larger quantum algorithm that creates and manipulates weights of superpositions.

\paragraph{\AllocatorName}

We present \AllocatorName{}, the first memory allocator that enables reversible, history-independent, and constant-time dynamic memory allocation in superposition.
A developer may invoke \AllocatorName{} to write a history-independent implementation of a data structure.

\paragraph{History Independence}

To leverage \AllocatorName{}, the developer constructs dynamically allocated, linked data whose linking structure is in one-to-one correspondence with semantic data structure instances. At runtime, \AllocatorName{} stores data into superpositions of allocation sites that are chosen to guarantee that the instance's physical memory representation, including the values of pointers, is unique.

Classically, a memory allocator may establish history independence by randomizing allocation addresses. The implementation of \AllocatorName{} lifts randomization to quantum superposition via \emph{symmetrization}, an approach from quantum simulation~\citep{abrams1997,Berry2017ImprovedTF}.

\paragraph{\LangName}

We combine \CoreLang{} and \AllocatorName{} to build \LangName{}, a programming language for quantum programming with data structures. \LangName{} features co-recursive types and recursive functions that enable the developer to recursively define data structures and operations.

\paragraph{Bounded-Time Execution}
The type system of \LangName{} ensures that all recursive functions are bounded in size by classical parameters. A \LangName{} program lowers to \CoreLang{} by unrolling recursive calls, meaning that it corresponds to a bounded-size quantum circuit.

\subsection{\LibraryName{}: Data Structure Library}

We use \LangName{} to implement \LibraryName{}, the first quantum library of data structures, consisting of lists, stacks, queues, strings, and sets.
Of note, we demonstrate the first executable implementation of sets that is reversible, history-independent, and bounded-time, using radix trees~\citep{morrison} as proposed by~\citet{bernstein2013}.
Prior to this work, we are not aware of any existing quantum implementation of the set data structure that satisfies all three properties.\footnote{Concurrent to this work, \citet{gidney2022} provides an implementation of a quantum dictionary data structure. Their implementation does not use random-access memory and achieves linear time complexity, whereas our implementation achieves poly-logarithmic time complexity. Their work was made public while this work was under conference review.}

As a case study, we describe three challenges we encountered during the development process that do not arise in classical programming, as well as techniques to overcome them:

\begin{itemize}
    \item \emph{Recursive uncomputation} means that functions that uncompute the results of recursive calls may incur exponential slowdown in time complexity. We instead structure a program to avoid having recursive calls be temporary values.
    \item \emph{Mutated uncomputation} means that directly uncomputing temporary variables after the program state has been mutated may result in incorrect outcomes. We instead correctly uncompute values using context from the updated program state.
    \item \emph{Branch sequentialization} means that the time complexity of a conditionally branching quantum program is equal to the sum and not the maximum of its branches. To prevent inefficiency, we structure a program to avoid recursion under branches.
\end{itemize}

\subsection{Contributions}

In this work, we present the following contributions:
\begin{itemize}
    \item \emph{\CoreLang{}} (\Cref{sec:core-semantics}). We present \CoreLang{}, the first language for quantum programming with random-access memory. We formalize its type system, reversible operational semantics, program inversion transformation, and translation into a unitary quantum circuit.
    \item \emph{\AllocatorName{}} (\Cref{sec:memory-allocation}). We present \AllocatorName{}, the first allocator that enables reversible, history-independent, and constant-time dynamic memory allocation in superposition. We present the implementation of \AllocatorName{} using techniques from quantum simulation.
    \item \emph{\LangName{}} (\Cref{sec:main-semantics}). We present \LangName{}, a language that enables bounded-time quantum programming with recursively defined data structures. We formalize its type system, which verifies that recursion is classically bounded, and its lowering into \CoreLang{}.
    \item \emph{\LibraryName{}} (\Cref{sec:data-structure-library}). We present \LibraryName{}, the first quantum library of abstract data structures, including lists, stacks, queues, strings, and sets. We present the first executable set implementation satisfying reversibility, history independence, and bounded-time execution.
    \item \emph{Case Study} (\Crefrange{sec:recursive-uncomputation}{sec:branch-sequentialization}). We characterize three challenges we faced in the library implementation that do not arise in classical programming -- recursive uncomputation, mutated uncomputation, and branch sequentialization -- and techniques to overcome them.
\end{itemize}

\paragraph{Summary}

Our work establishes programming abstractions for data structures in quantum superposition in the form of language design and provides the first quantum library of fundamental data structures. Together, \LangName{}, \AllocatorName{}, and \LibraryName{} enable developers to implement algorithms and construct quantum software utilizing familiar and useful data structures.

\section{Background on Quantum Computation} \label{sec:background}

The following is an overview of key concepts in quantum computation relevant to this work. A comprehensive reference may be found in~\citet{nielsen_chuang_2010}.

\paragraph{Qubit}
The basic unit of quantum information is the \emph{qubit}, a linear combination $\gamma_0 \ket{0} + \gamma_1 \ket{1}$ known as a \emph{superposition}, where $\ket{0}$ and $\ket{1}$ are \emph{basis states} and $\gamma_0, \gamma_1 \in \mathbb{C}$ are \emph{amplitudes} satisfying $\abs{\gamma_0}^2 + \abs{\gamma_1}^2 = 1$ describing relative weights of basis states. Examples of qubits include the classical zero bit $\ket{0}$, classical one bit $\ket{1}$, and the superposition states $\frac{1}{\sqrt{2}}{(\ket{0} + \ket{1})}$ and $\frac{1}{\sqrt{2}}{(i\ket{0} - \ket{1})}$.

\paragraph{Quantum State}
A $2^n$-dimensional \emph{quantum state} $\ket{\qstate}$ is a superposition over $n$-bit strings. For example, $\frac{1}{\sqrt{2}}{(\ket{00}+\ket{11})}$ is a quantum state over two qubits.
Multiple qubits form a quantum state system by means of the tensor product $\otimes$, e.g. the state $\ket{01}$ is equal to the product $\ket{0} \otimes \ket{1}$. As is standard in quantum computing, we also use the notation $\ket{x, y}$ to represent $\ket{x} \otimes \ket{y}$.

\paragraph{Unitary Operators}
A \emph{unitary operator} $U$ is a linear operator on quantum states that preserves inner products and whose inverse is its Hermitian adjoint.
Examples of unitary operators include single-qubit \emph{quantum gates} such as:
\begin{itemize}
    \item X --- bit-flip (NOT) gate, which maps $\ket{0} \mapsto \ket{1}$ and vice versa;
    \item Z --- phase-flip gate, which leaves $\ket{0}$ unchanged and maps $\ket{1} \mapsto -\ket{1}$;
    \item H --- Hadamard gate, which maps $\ket{0} \mapsto \frac{1}{\sqrt{2}}{(\ket{0} + \ket{1})}$ and $\ket{1} \mapsto \frac{1}{\sqrt{2}}{(\ket{0} - \ket{1})}$.
\end{itemize}
Other examples include two-qubit gates, such as controlled-NOT, controlled-Z, and SWAP. The controlled gates perform NOT or Z on their target qubit if their control qubit is in state $\ket{1}$, and the SWAP gate swaps two qubits in a quantum state.

\paragraph{Measurement}
A quantum \emph{measurement} is a probabilistic operation over quantum states. When a qubit $\gamma_0 \ket{0} + \gamma_1 \ket{1}$ is measured,\footnote{In this work, we refer only to projective computational (i.e. 0/1) basis measurements.} the outcome is $\ket{0}$ with probability $\abs{\gamma_0}^2$ and $\ket{1}$ with probability $\abs{\gamma_1}^2$. Measuring a qubit within a quantum state causes the entire state to probabilistically assume one of two outcome states.
To define the outcomes of measuring the first qubit of $\ket{\qstate}$, we first rewrite the state into the form $\ket{\qstate} = \gamma_0 \ket{0} \otimes \ket{\qstate_0}+ \gamma_1 \ket{1} \otimes \ket{\qstate_1}$. Then, with probability $\abs{\gamma_0}^2$ the measurement outcome is $\ket{0} \otimes \ket{\qstate_0}$, and with probability $\abs{\gamma_1}^2$ the outcome is $\ket{1} \otimes \ket{\qstate_1}$.
Outcomes of measuring multiple qubits or qubits other than the first are defined analogously.

\paragraph{Random Access}

Quantum random-access memory (qRAM)\footnotemark{}~\citep{Giovannetti_2008} enables manipulating data at addresses in superposition.
\footnotetext{The abbreviation qRAM is not to be confused with the QRAM model of computation by~\citet{knill}, which does not incorporate quantum random access.}
Let $k$ be the system word size.
The single-qubit \emph{quantum random-access gate} is defined to be the unitary operation that maps~\citep{ambainis2003}:
\[
\ket{i, b, z_1, \ldots, z_m} \mapsto \ket{i, z_i, z_1, \ldots, z_{i-1}, b, z_{i+1}, \ldots, z_m}
\]
where $i$ is a $k$-bit integer, $b$ is a bit, and $\ket{z_i, \ldots, z_m}$ is the qRAM --- an array of $m$ qubits. The effect of this gate is to swap the data at position $i$ in the array $z$ with the data in the register $b$.

In this work, we utilize a multi-qubit random-access gate, in which $b$ and each $z_i$ is a $k$-bit string rather than a single bit. Such a gate may be implemented as a straightforward extension of leading proposals for hardware qRAM~\citep{Paler_2020,Arunachalam_2015,Matteo_2020}.

We account for the cost of data structure operations using a model in which quantum random access takes $O(1)$ time. This accounting facilitates a direct comparison to the classical RAM cost model, in which dereferencing a pointer is modeled as constant-time regardless of the underlying hardware design. It is possible to use an alternative cost model in which quantum random access takes $O(\log m)$ time, which is used by~\citet{ambainis2003,aaronson2019,bernstein2013,buhrman2021}. To account under this model, the reported complexity of each data structure operation should be multiplied by a poly-logarithmic time factor.

\section{Quantum Data Structures in \LangName{}} \label{sec:examples}

In this section, we present data structures in quantum superposition and justify the properties of reversibility, history independence, and bounded-time execution. We illustrate how a developer uses \LangName{} to implement fundamental data structures satisfying these properties.

\subsection{Linked Data in Superposition}

In classical computing, linked data forms the basis of numerous abstract data structures. For example, linked lists may be used to concretely implement abstract stacks, queues, and sets.
We use the notation \pstate{\texttt{l} \goesto [1,2]} to represent a classical program state\footnotemark{} in which the variable \lstinline{l} denotes a linked list storing the integers 1 and 2.
\footnotetext{A classical program state is a bit string, meaning that the set of all program states is finite.}

\paragraph{Superposition}

In a \LangName{} program, the state of all program variables and memory exists in quantum superposition. For example, the quantum state of the program may be in a uniform superposition of three classical states in which \lstinline{l} denotes different linked lists:
\[
    \frac{1}{\sqrt{3}}\left(\ket{\pstate{\texttt{l} \goesto []}} + \ket{\pstate{\texttt{l} \goesto [1,2]}} + \ket{\pstate{\texttt{l} \goesto [1,\ldots,100]}}\right)
\]

\paragraph{Data Structure Operations}
We may lift a classical operation on a list to a unitary operation on a list in superposition. For example, the operation \lstinline{push_front} prepends an element onto every list in the superposition. Executing $\texttt{push\_front}(\texttt{l}, 6)$ would produce a new quantum state:
\[
    \frac{1}{\sqrt{3}}\left(\ket{\pstate{\texttt{l} \goesto [6]}} + \ket{\pstate{\texttt{l} \goesto [6,1,2]}} + \ket{\pstate{\texttt{l} \goesto [6,1,\ldots,100]}}\right)
\]

\paragraph{Memory Representation}

\newcommand{\nullbox}{\,\raisebox{-0.1em}{\scalebox{0.8}{\begin{tikzpicture}
\draw (2.4, 0.6) rectangle (2.8, 1);
\draw (2.4, 0.6) -- (2.8, 1);
\end{tikzpicture}}}\,}

The memory representation of a singly-linked list is a collection of nodes, each consisting of content data and a pointer to the next node, or \lstinline{null} if there is none.

Within a classical program state, \lstinline{null} denotes the empty list, and a pointer to a memory location at which a pair of 1 and \lstinline{null} are stored denotes the singleton list containing the element 1. The program state $\pstate{\texttt{l} \goesto [1,2]}$ is depicted by the following diagram, in which \nullbox{} denotes \lstinline{null}:
\vspace{0.5em}
\begin{center}
\begin{tikzpicture}
\node at (-0.3, 0.8) {\texttt{l}};
\draw (0, 0.6) rectangle (0.4, 1);
\draw[black,fill=black] (0.2, 0.8) circle (.4ex);
\draw[-stealth, thick] (0.2, 0.8) -- (0.75, 0.8);

\draw (0.8, 0.6) rectangle (1.2, 1);
\draw (1.2, 0.6) rectangle (1.6, 1);
\node at (1, 0.8) {1};
\draw[black,fill=black] (1.4, 0.8) circle (.4ex);
\draw[-stealth, thick] (1.4, 0.8) -- (1.95, 0.8);

\draw (2, 0.6) rectangle (2.4, 1);
\draw (2.4, 0.6) rectangle (2.8, 1);
\node at (2.2, 0.8) {2};
\draw (2.4, 0.6) -- (2.8, 1);
\end{tikzpicture}
\end{center}
\subsection{Implementing Linked Data Structures in \LangName{}}

\LangName{} makes it possible to implement linked data structures in quantum superposition. A key aspect of the design of \LangName{} is a set of programming primitives that are each \emph{reversible}, meaning a program can be executed forwards or backwards to produce or reverse its computational effect.

\paragraph{Reversibility}

Reversible primitives correspond to unitary operators that can be realized as gates in a quantum computer (\Cref{sec:background}). Though a quantum computer is also capable of performing irreversible measurement, measurement collapses the program state from superposition. Algorithms such as~\citet{ambainis2003,aaronson2019,bernstein2013} are therefore designed to avoid performing measurement on the data structure in superposition, which if performed would cause the algorithm to produce an incorrect output or lose its efficiency advantage.

Correspondingly, we designed \LangName{} with reversible primitives that make it possible to manipulate data structures without collapsing their superposition.

\begin{figure}
\centering
\vspace*{-.5em}
\begin{minipage}[t]{.5\textwidth}
\begin{lstlisting}[label={lst:push-front}]
type list = (uint, ptr<list>);
fun push_front(l: ptr<list>, x: uint) {
  let head <- alloc<list>;
  l <-> head;                /* Swap */
  let node <- (x, head);
  let head -> node.2;        /* Uncompute */
  *l <-> node;               /* Swap */
  let node -> default<list>; /* Uncompute */
  return ();
}
\end{lstlisting}
\end{minipage}%
\begin{minipage}[t]{.5\textwidth}
\begin{lstlisting}[label={lst:pop-front}]
fun pop_front(l: ptr<list>) -> uint {
  let node <- default<list>;
  *l <-> node;
  let head <- node.2;
  let x <- node.1;  /* Get output */
  let node -> (x, head);
  l <-> head;
  let head -> alloc<list>;
  return x;
}
\end{lstlisting}
\end{minipage}

\setlength{\abovecaptionskip}{0pt}
\setlength{\belowcaptionskip}{-1em}
\begin{minipage}[b]{0.5\textwidth}
\captionof{figure}{Implementation of \lstinline{push_front}.} \label{fig:push-front}
\end{minipage}%
\begin{minipage}[b]{0.5\textwidth}
\captionof{figure}{Implementation of \lstinline{pop_front}.} \label{fig:pop-front}
\end{minipage}
\end{figure}

\paragraph{Implementing \texttt{push\_front}}

In~\Cref{fig:push-front}, we present the \LangName{} implementation of the operation $\texttt{push\_front}$ on linked lists. When executed forwards, $\texttt{push\_front}(\texttt{l}, 6)$  prepends 6 to the list \lstinline{l}, and when executed backwards, it removes 6 from the front of \lstinline{l}. At a high level, the function initializes a new list node that stores \lstinline{x} and points to the input list, and then points \lstinline{l} to this node.

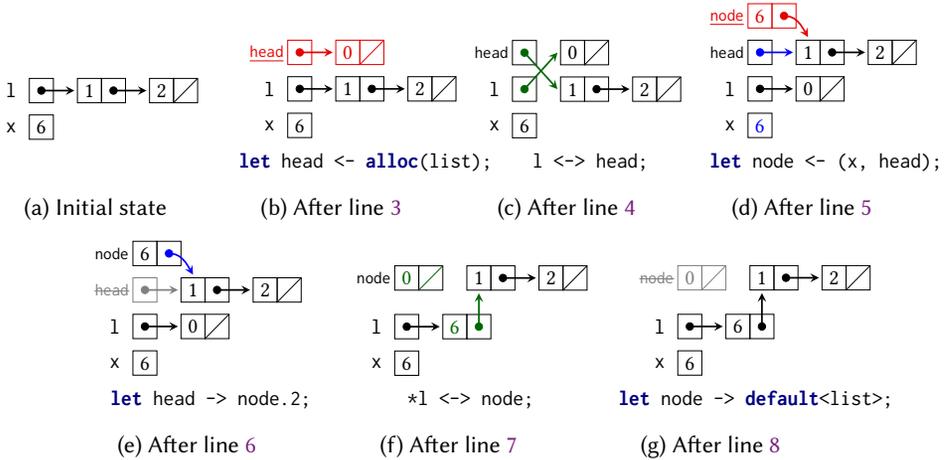
\begin{figure}
\centering
\begin{subfigure}[t]{.2\textwidth}
\centering
\scalebox{0.8}{
\begin{tikzpicture}
\node at (1.3, -0.4) {\vphantom{l}};
\node at (-0.3, 0.2) {\texttt{x}};
\draw (0, 0) rectangle (0.4,0.4);
\node at (0.2, 0.2) {6};

\node at (-0.3, 0.8) {\texttt{l}};
\draw (0, 0.6) rectangle (0.4, 1);
\draw[black,fill=black] (0.2, 0.8) circle (.4ex);
\draw[-stealth, thick] (0.2, 0.8) -- (0.75, 0.8);

\draw (0.8, 0.6) rectangle (1.2, 1);
\draw (1.2, 0.6) rectangle (1.6, 1);
\node at (1, 0.8) {1};
\draw[black,fill=black] (1.4, 0.8) circle (.4ex);
\draw[-stealth, thick] (1.4, 0.8) -- (1.95, 0.8);

\draw (2, 0.6) rectangle (2.4, 1);
\draw (2.4, 0.6) rectangle (2.8, 1);
\node at (2.2, 0.8) {2};
\draw (2.4, 0.6) -- (2.8, 1);
\end{tikzpicture}}
\caption{Initial state}
\label{fig:push-step-1}
\end{subfigure}\quad%
\begin{subfigure}[t]{.2\textwidth}
\centering
\scalebox{0.8}{
\begin{tikzpicture}
\node at (1.3, -0.4) {\lstinline[style=color]{let head <- alloc(list);}};
\node at (-0.3, 0.2) {\texttt{x}};
\draw (0, 0) rectangle (0.4,0.4);
\node at (0.2, 0.2) {6};

\node at (-0.3, 0.8) {\texttt{l}};
\draw (0, 0.6) rectangle (0.4, 1);
\draw[black,fill=black] (0.2, 0.8) circle (.4ex);
\draw[-stealth, thick] (0.2, 0.8) -- (0.75, 0.8);

\draw (0.8, 0.6) rectangle (1.2, 1);
\draw (1.2, 0.6) rectangle (1.6, 1);
\node at (1, 0.8) {1};
\draw[black,fill=black] (1.4, 0.8) circle (.4ex);
\draw[-stealth, thick] (1.4, 0.8) -- (1.95, 0.8);

\draw (2, 0.6) rectangle (2.4, 1);
\draw (2.4, 0.6) rectangle (2.8, 1);
\node at (2.2, 0.8) {2};
\draw (2.4, 0.6) -- (2.8, 1);

\node[red!90!black] at (-0.35, 1.37) {\footnotesize \underline{\texttt{head}}};
\draw[red!90!black] (0, 1.2) rectangle (0.4, 1.6);
\draw[red!90!black, fill=red!90!black] (0.2, 1.4) circle (.4ex);
\draw[-stealth, thick, red!90!black] (0.2, 1.4) -- (0.75, 1.4);

\draw[red!90!black] (0.8, 1.2) rectangle (1.2, 1.6);
\draw[red!90!black] (1.2, 1.2) rectangle (1.6, 1.6);
\node[red!90!black] at (1, 1.4) {0};
\draw[red!90!black] (1.2, 1.2) -- (1.6, 1.6);
\end{tikzpicture}}
\caption{After line~\ref{lst:push-front-3}}
\label{fig:push-step-2}
\end{subfigure}\quad%
\begin{subfigure}[t]{.2\textwidth}
\centering
\scalebox{0.8}{
\begin{tikzpicture}
\node at (1.3, -0.4) {\lstinline[style=color]{l <-> head;}};
\node at (-0.3, 0.2) {\texttt{x}};
\draw (0, 0) rectangle (0.4,0.4);
\node at (0.2, 0.2) {6};

\node at (-0.3, 0.8) {\texttt{l}};
\draw (0, 0.6) rectangle (0.4, 1);
\draw[green!40!black, fill=green!40!black] (0.2, 0.8) circle (.4ex);
\draw[-stealth, thick, green!40!black] (0.2, 0.8) -- (0.75, 1.4);

\draw (0.8, 0.6) rectangle (1.2, 1);
\draw (1.2, 0.6) rectangle (1.6, 1);
\node at (1, 0.8) {1};
\draw[black,fill=black] (1.4, 0.8) circle (.4ex);
\draw[-stealth, thick] (1.4, 0.8) -- (1.95, 0.8);

\draw (2, 0.6) rectangle (2.4, 1);
\draw (2.4, 0.6) rectangle (2.8, 1);
\node at (2.2, 0.8) {2};
\draw (2.4, 0.6) -- (2.8, 1);

\node at (-0.35, 1.4) {\footnotesize \texttt{head}};
\draw (0, 1.2) rectangle (0.4, 1.6);
\draw[green!40!black, fill=green!40!black] (0.2, 1.4) circle (.4ex);
\draw[-stealth, thick, green!40!black] (0.2, 1.4) -- (0.75, 0.8);

\draw (0.8, 1.2) rectangle (1.2, 1.6);
\draw (1.2, 1.2) rectangle (1.6, 1.6);
\node at (1, 1.4) {0};
\draw (1.2, 1.2) -- (1.6, 1.6);
\end{tikzpicture}}
\caption{After line~\ref{lst:push-front-4}}
\label{fig:push-step-3}
\end{subfigure}\quad%
\begin{subfigure}[t]{.2\textwidth}
\centering
\scalebox{0.8}{
\begin{tikzpicture}
\node at (1.3, -0.4) {\lstinline[style=color]{let node <- (x, head);}};
\node at (-0.3, 0.2) {\texttt{x}};
\draw (0, 0) rectangle (0.4,0.4);
\node[blue] at (0.2, 0.2) {6};

\node at (-0.3, 0.8) {\texttt{l}};
\draw (0, 0.6) rectangle (0.4, 1);
\draw[black, fill=black] (0.2, 0.8) circle (.4ex);
\draw[-stealth, thick] (0.2, 0.8) -- (0.75, 0.8);

\draw (0.8, 1.2) rectangle (1.2, 1.6);
\draw (1.2, 1.2) rectangle (1.6, 1.6);
\node at (1, 1.4) {1};
\draw[black,fill=black] (1.4, 1.4) circle (.4ex);
\draw[-stealth, thick] (1.4, 1.4) -- (1.95, 1.4);

\draw (2, 1.2) rectangle (2.4, 1.6);
\draw (2.4, 1.2) rectangle (2.8, 1.6);
\node at (2.2, 1.4) {2};
\draw (2.4, 1.2) -- (2.8, 1.6);

\node at (-0.35, 1.4) {\footnotesize \texttt{head}};
\draw (0, 1.2) rectangle (0.4, 1.6);
\draw[blue, fill=blue] (0.2, 1.4) circle (.4ex);
\draw[-stealth, thick, blue] (0.2, 1.4) -- (0.75, 1.4);

\draw (0.8, 0.6) rectangle (1.2, 1);
\draw (1.2, 0.6) rectangle (1.6, 1);
\node at (1, 0.8) {0};
\draw (1.2, 0.6) -- (1.6, 1);

\node[red!90!black] at (-0.35, 1.97) {\footnotesize \underline{\texttt{node}}};
\draw[red!90!black] (0, 1.8) rectangle (0.4, 2.2);
\node[red!90!black] at (0.2, 2) {6};
\draw[red!90!black] (0.4, 1.8) rectangle (0.8, 2.2);
\draw[red!90!black, fill=red!90!black] (0.6, 2) circle (.4ex);
\draw[-stealth, thick, red!90!black] (0.6, 2) to[out=0,in=120] (1, 1.65);
\end{tikzpicture}}
\caption{After line~\ref{lst:push-front-5}}
\label{fig:push-step-4}
\end{subfigure}

\vspace{5pt}

\begin{subfigure}[t]{.2\textwidth}
\centering
\scalebox{0.8}{
\begin{tikzpicture}
\node at (1.3, -0.4) {\lstinline[style=color]{let head -> node.2;}};
\node at (-0.3, 0.2) {\texttt{x}};
\draw (0, 0) rectangle (0.4,0.4);
\node at (0.2, 0.2) {6};

\node at (-0.3, 0.8) {\texttt{l}};
\draw (0, 0.6) rectangle (0.4, 1);
\draw[black, fill=black] (0.2, 0.8) circle (.4ex);
\draw[-stealth, thick] (0.2, 0.8) -- (0.75, 0.8);

\node[gray] at (-0.35, 1.4) {\footnotesize \sout{\texttt{head}}};
\draw[gray] (0, 1.2) rectangle (0.4, 1.6);
\draw[gray, fill=gray] (0.2, 1.4) circle (.4ex);
\draw[-stealth, thick, gray] (0.2, 1.4) -- (0.75, 1.4);

\draw (0.8, 0.6) rectangle (1.2, 1);
\draw (1.2, 0.6) rectangle (1.6, 1);
\node at (1, 0.8) {0};
\draw (1.2, 0.6) -- (1.6, 1);

\node at (-0.35, 2) {\footnotesize \texttt{node}};
\draw (0, 1.8) rectangle (0.4, 2.2);
\node at (0.2, 2) {6};
\draw (0.4, 1.8) rectangle (0.8, 2.2);
\draw[blue, fill=blue] (0.6, 2) circle (.4ex);
\draw[-stealth, thick, blue] (0.6, 2) to[out=0,in=120] (1, 1.65);

\draw (0.8, 1.2) rectangle (1.2, 1.6);
\draw (1.2, 1.2) rectangle (1.6, 1.6);
\node at (1, 1.4) {1};
\draw[black,fill=black] (1.4, 1.4) circle (.4ex);
\draw[-stealth, thick] (1.4, 1.4) -- (1.95, 1.4);

\draw (2, 1.2) rectangle (2.4, 1.6);
\draw (2.4, 1.2) rectangle (2.8, 1.6);
\node at (2.2, 1.4) {2};
\draw (2.4, 1.2) -- (2.8, 1.6);
\end{tikzpicture}}
\caption{After line~\ref{lst:push-front-6}}
\label{fig:push-step-5}
\end{subfigure}\qquad%
\begin{subfigure}[t]{.2\textwidth}
\centering
\scalebox{0.8}{
\begin{tikzpicture}
\node at (1.3, -0.4) {\lstinline[style=color]{*l <-> node;}};
\node at (-0.3, 0.2) {\texttt{x}};
\draw (0, 0) rectangle (0.4,0.4);
\node at (0.2, 0.2) {6};

\node at (-0.3, 0.8) {\texttt{l}};
\draw (0, 0.6) rectangle (0.4, 1);
\draw[black, fill=black] (0.2, 0.8) circle (.4ex);
\draw[-stealth, thick] (0.2, 0.8) -- (0.75, 0.8);

\draw (0.8, 0.6) rectangle (1.2, 1);
\draw (1.2, 0.6) rectangle (1.6, 1);
\node[green!40!black] at (1, 0.8) {6};
\draw[green!40!black, fill=green!40!black] (1.4, 0.8) circle (.4ex);
\draw[-stealth, thick, green!40!black] (1.4, 0.8) -- (1.4, 1.35);

\node at (-0.35, 1.6) {\footnotesize \texttt{node}};
\draw (0, 1.4) rectangle (0.4, 1.8);
\node[green!40!black] at (0.2, 1.6) {0};
\draw (0.4, 1.4) rectangle (0.8, 1.8);
\draw[green!40!black] (0.4, 1.4) -- (0.8, 1.8);

\draw (1.2, 1.4) rectangle (1.6, 1.8);
\draw (1.6, 1.4) rectangle (2, 1.8);
\node at (1.4, 1.6) {1};
\draw[black,fill=black] (1.8, 1.6) circle (.4ex);
\draw[-stealth, thick] (1.8, 1.6) -- (2.35, 1.6);

\draw (2.4, 1.4) rectangle (2.8, 1.8);
\draw (2.8, 1.4) rectangle (3.2, 1.8);
\node at (2.6, 1.6) {2};
\draw (2.8, 1.4) -- (3.2, 1.8);
\end{tikzpicture}}
\caption{After line~\ref{lst:push-front-7}}
\label{fig:push-step-6}
\end{subfigure}\qquad%
\begin{subfigure}[t]{.2\textwidth}
\centering
\scalebox{0.8}{
\begin{tikzpicture}
\node at (1.3, -0.4) {\lstinline[style=color]{let node -> default<list>;}};
\node at (-0.3, 0.2) {\texttt{x}};
\draw (0, 0) rectangle (0.4,0.4);
\node at (0.2, 0.2) {6};

\node at (-0.3, 0.8) {\texttt{l}};
\draw (0, 0.6) rectangle (0.4, 1);
\draw[black, fill=black] (0.2, 0.8) circle (.4ex);
\draw[-stealth, thick] (0.2, 0.8) -- (0.75, 0.8);

\draw (0.8, 0.6) rectangle (1.2, 1);
\draw (1.2, 0.6) rectangle (1.6, 1);
\node at (1, 0.8) {6};
\draw[black, fill=black] (1.4, 0.8) circle (.4ex);
\draw[-stealth, thick] (1.4, 0.8) -- (1.4, 1.35);

\node[gray] at (-0.35, 1.6) {\footnotesize \sout{\texttt{node}}};
\draw[gray] (0, 1.4) rectangle (0.4, 1.8);
\node[gray] at (0.2, 1.6) {0};
\draw[gray] (0.4, 1.4) rectangle (0.8, 1.8);
\draw[gray] (0.4, 1.4) -- (0.8, 1.8);

\draw (1.2, 1.4) rectangle (1.6, 1.8);
\draw (1.6, 1.4) rectangle (2, 1.8);
\node at (1.4, 1.6) {1};
\draw[black,fill=black] (1.8, 1.6) circle (.4ex);
\draw[-stealth, thick] (1.8, 1.6) -- (2.35, 1.6);

\draw (2.4, 1.4) rectangle (2.8, 1.8);
\draw (2.8, 1.4) rectangle (3.2, 1.8);
\node at (2.6, 1.6) {2};
\draw (2.8, 1.4) -- (3.2, 1.8);
\end{tikzpicture}}
\caption{After line~\ref{lst:push-front-8}}
\label{fig:push-step-7}
\end{subfigure}
\caption{Program states after executing each line of $\texttt{push\_front}(\texttt{l}, 6)$, where \texttt{l} denotes $[1,2]$. In each step, newly introduced features are shown in \textcolor{red!90!black}{red}, removed features are shown in \textcolor{gray}{gray}, and swapped features are shown in \textcolor{green!40!black}{green}. Features used for computing or uncomputing other features are shown in \textcolor{blue}{blue}.} \label{fig:push-steps}
\end{figure}

In~\Cref{fig:push-steps}, we depict the effect of each line of the procedure diagrammatically.
In~\Cref{fig:push-step-1}, we show the initial program state of $\texttt{push\_front}(\texttt{l},6)$, assuming \lstinline{l} denotes $[1,2]$. In~\Cref{fig:push-step-2}, we depict the effect of line~\ref{lst:push-front-3}, which uses the \emph{assignment} operator \lstinline{<-} to assign the variable \lstinline{head} to a new list node in memory, allocated using the \lstinline{alloc} operator and initialized to an empty default value. Line~\ref{lst:push-front-4} (\Cref{fig:push-step-3}) uses the \emph{swap} operator \lstinline{<->} to reversibly swap the values of \lstinline{head} and \lstinline{l}. Line~\ref{lst:push-front-5} (\Cref{fig:push-step-4}) assigns \lstinline{node} to be a pair of the input \lstinline{x} and \lstinline{head}.

\paragraph{Uncomputation}
At this point, the variable \lstinline{head} is no longer useful.
Line~\ref{lst:push-front-6} (\Cref{fig:push-step-5}) uses the \emph{un-assignment} operator \lstinline{->} to \emph{un-assign} \lstinline{head} from the second field of \lstinline{node}, meaning that it \emph{uncomputes} the value of the variable \lstinline{head} and then discards the variable from the program.

Uncomputation~\citep{bennett1973} is the action of reversing the computation that produces the value of a variable, restoring that variable's value to zero.
The reason uncomputation is necessary here is that by the design of \LangName{}, every program is reversible, meaning that 1) a program may only discard a variable after restoring its value to zero, and 2) no language primitive enables the program to generally erase any unspecified value to zero. If the forward execution of the program were permitted to discard or erase an unspecified value, then the reverse execution would be forced to correctly re-materialize that discarded or erased value, which is impossible.

On line~\ref{lst:push-front-6}, the program uses the un-assignment operator to restore the value of \lstinline{head} to zero.\footnote{\LangName{} also supports other operators that abstract the management of temporary values, which we discuss in~\Cref{sec:main-semantics}. We focus here on the un-assignment operator for explanatory purposes.} Specifically, the effect of line~\ref{lst:push-front-6} is defined as the reverse of that of the statement:
\begin{lstlisting}[style=nonum]
let head <- node.2;
\end{lstlisting}
An assignment such as this one may be interpreted as a reversible sequence of steps. First, \lstinline{head} is initialized to zero. Then, each bit of \lstinline{head} is flipped if and only if the corresponding bit of \lstinline{node.2} is set. The result is that the values of \lstinline{head} and \lstinline{node.2} are equal.

The un-assignment on line~\ref{lst:push-front-6} may be interpreted as the above sequence of steps in reverse. It starts with \lstinline{head} and \lstinline{node.2} having equal values, and performs the same bit flips on \lstinline{head} conditioned on \lstinline{node.2}. The result is that \lstinline{head} has a value of zero, which the program safely discards.

Continuing on, line~\ref{lst:push-front-7} (\Cref{fig:push-step-6}) swaps \lstinline{node} with the contents of memory at address \lstinline{l}. Recalling that after line~\ref{lst:push-front-4}, \lstinline{l} stores an empty default node, it is then the case that after line~\ref{lst:push-front-7}, \lstinline{l} points to the new list head and \lstinline{node} stores an empty node, which is then uncomputed by line~\ref{lst:push-front-8} (\Cref{fig:push-step-7}).

\paragraph{Reversing Operation}

The reversibility of \LangName{} programs makes it simple to implement the closely related operation of \lstinline{pop_front}, which removes the first element from a list and returns the element. Indeed, the implementation of \lstinline{pop_front} is nearly the syntactic reverse of \lstinline{push_front}.

In~\Cref{fig:pop-front}, we present \lstinline{pop_front}. It first reverses line~\ref{lst:push-front-8} of \lstinline{push_front}, replacing un-assignment with assignment, and then line~\ref{lst:push-front-7}, which remains identical --- the reverse of a swap operation is itself. It reverses line~\ref{lst:push-front-6} of \lstinline{push_front}, assigning the second field of \lstinline{node} to \lstinline{head}. Line~\ref{lst:pop-front-5} of \lstinline{pop_front} is new to this function --- it obtains \lstinline{x} from the first field of \lstinline{node} so that \lstinline{x} may be returned. The following lines reverse the first three lines of \lstinline{push_front}, and the function returns \lstinline{x}.

\subsection{Dynamic Memory Allocation with \AllocatorName{}}

Sets are foundational to programming and critical to algorithms. Surprisingly, many implementations of sets such as lists, hash tables, and binary trees are unsuitable for quantum programming.

\paragraph{Interference} \label{sec:interference}

Algorithms such as~\citet{ambainis2003,aaronson2019,bernstein2013,buhrman2021} rely on \emph{interference}, the constructive or destructive interaction of amplitudes in a quantum state.
For example,~\citet{ambainis2003} generates a superposition over program states $\pstate{\texttt{s} \goesto s_i}$, meaning that variable \texttt{s} stores a set of integers $s_i$. The correctness of the algorithm relies on the fact that if $s_1$ and $s_2$ denote sets containing the same integers, then their amplitudes interfere. In particular, suppose that $s_1$ and $s_2$ both denote the set $\{1, 2\}$. Then, the algorithm requires that:
\[
    \gamma \ket{\pstate{\texttt{s} \goesto s_1}} - \gamma \ket{\pstate{\texttt{s} \goesto s_2}} = 0
\]

\paragraph{History Independence}

However, suppose that sets are implemented using linked lists, where $s_1$ is represented as $[1,2]$ and $s_2$ as $[2,1]$. Because $\pstate{\texttt{s} \goesto [1,2]}$ and $\pstate{\texttt{s} \goesto [2,1]}$ are physically distinct states, these terms do not cancel, causing the algorithm to produce incorrect outputs.

For the algorithm to produce correct outputs, the implementation of the data structure must satisfy the property of \emph{history independence}.
We adapt its definition from \citet{teague2001}:\footnotemark{}
\footnotetext{This definition is termed the ``strong'' variant of history independence in the original work. Conceptually similar definitions are found in cryptographic models of computing, such as~\citet{wang2014}.}
\begin{definition}[History Independence]
    An implementation of an abstract data structure is \emph{history-independent} if and only if, when any two sequences of abstract operations yield semantically equivalent data structure instances, then the physical representations of the instances are identical.
\end{definition}

The implementation of a set as an unsorted linked list is not history-independent, as the lists $[1,2]$ and $[2,1]$ both represent the same set $\{1,2\}$.
Similarly, implementations of sets such as balanced binary trees or hash tables\footnotemark{} are not history-independent and are thus unsuitable for use in quantum algorithms that leverage interference on data structures.
\footnotetext{Other reasons why hash tables are unsuitable are 1) removal from a hash table is not directly reversible, and 2) all operations assume their worst case $O(n)$ rather than average case $O(1)$ complexity when all recursion is classically unfolded.}

\paragraph{Unique Linking Structure}

A history-independent implementation must assign each semantic instance of the data structure a unique linking structure in the underlying representation.
For example, suppose a set is instead implemented as a sorted list without duplicates. The set $\{1,2\}$ is then uniquely represented as the list $[1,2]$, ostensibly satisfying history independence.

\paragraph{Memory Addresses}

However, a unique linking structure does not by itself guarantee history independence.
In the example, the issue is that $\pstate{\texttt{s} \goesto [1,2]}$ may still refer to physically distinct program states in which the pointers to the linked list nodes take on different values.

\begin{figure}
\centering
\begin{minipage}[t]{0.54\textwidth}
\begin{subfigure}[t]{.5\textwidth}
\centering
\begin{tikzpicture}
\draw (0,0.7) rectangle (0.4,1.1);
\draw[black,fill=black](0.2, 0.9) circle (.4ex);
\draw[-stealth, thick] (0.2, 0.9) to (0.2, 0.42);
\node at (-0.25, 0.9) {\texttt{s}};

\draw (0,0) rectangle (0.4,0.4);
\draw (0.4,0) rectangle (0.8,0.4);
\draw (0.8,0) rectangle (1.2,0.4);
\draw (1.2,0) rectangle (1.6,0.4);
\draw[fill=gray!30] (1.6,0) rectangle (2,0.4);
\draw[fill=gray!30] (2,0) rectangle (2.4,0.4);

\node at (0.2,0.2) {$1$};
\draw[black,fill=black] (0.6,0.2) circle (.4ex);
\node at (1,0.2) {$2$};
\draw (1.2,0) -- (1.6,0.4);

\draw[-stealth, thick] (0.6,0.2) to[out=90,in=120,distance=0.4cm] (1,0.42);
\end{tikzpicture}
\captionof{figure}{One possible memory representation of $\pstate{\texttt{s} \goesto [1,2]}$.}
\label{fig:state-rep-1}
\end{subfigure}\hspace{1em}%
\begin{subfigure}[t]{.4\textwidth}
\centering
\begin{tikzpicture}
\draw (0,0.7) rectangle (0.4,1.1);
\draw[black,fill=black](0.2, 0.9) circle (.4ex);
\draw[-stealth, thick] (0.2, 0.9) to[out=0,in=90,distance=0.6cm] (1.8, 0.42);
\node at (-0.25, 0.9) {\texttt{s}};

\draw (0,0) rectangle (0.4,0.4);
\draw (0.4,0) rectangle (0.8,0.4);
\draw[fill=gray!30] (0.8,0) rectangle (1.2,0.4);
\draw[fill=gray!30] (1.2,0) rectangle (1.6,0.4);
\draw (1.6,0) rectangle (2,0.4);
\draw (2,0) rectangle (2.4,0.4);

\node at (1.8,0.2) {$1$};
\draw[black,fill=black] (2.2,0.2) circle (.4ex);
\node at (0.2,0.2) {$2$};
\draw (0.4,0) -- (0.8,0.4);
\draw[-stealth, thick] (2.2,0.2) to[out=80,in=60,distance=0.55cm] (0.2,0.42);
\end{tikzpicture}
\captionof{figure}{A distinct representation of the same state.}
\label{fig:state-rep-2}
\end{subfigure}

\vspace{0.5em}

\begin{subfigure}[t]{\textwidth}
\centering
\begin{tikzpicture}
\begin{scope}[yshift=1.75cm]
\node at (-1.6, 0.2) {$\dfrac{1}{\textcolor{gray}{\sqrt{6}}}$};

\draw[thick] (-1.2,-0.3) -- (-1.2, 0.7);
\draw[thick] (2.6,-0.3) -- (2.7, 0.2);
\draw[thick] (2.7,0.2) -- (2.6, 0.7);

\draw (-0.7, 0) rectangle (-0.3, 0.4);
\draw (0,0) rectangle (0.4,0.4);
\draw (0.4,0) rectangle (0.8,0.4);
\draw (0.8,0) rectangle (1.2,0.4);
\draw (1.2,0) rectangle (1.6,0.4);
\draw[fill=gray!30] (1.6,0) rectangle (2,0.4);
\draw[fill=gray!30] (2,0) rectangle (2.4,0.4);

\node at (-0.95, 0.2) {\texttt{s}};
\draw[black,fill=black] (-0.5,0.2) circle (.4ex);
\node at (0.2,0.2) {$1$};
\draw[black,fill=black] (0.6,0.2) circle (.4ex);
\node at (1,0.2) {$2$};
\draw (1.2,0) -- (1.6,0.4);
\draw[-stealth, thick] (0.6,0.2) to[out=90,in=120,distance=0.4cm] (1,0.42);
\draw[-stealth, thick] (-0.5,0.2) to (-0.02,0.2);
\end{scope}

\node at (0.7, 1.1) {$+\ \cdots\ + $};

\node at (-1.6, 0.2) {$\dfrac{1}{\textcolor{gray}{\sqrt{6}}}$};

\draw[thick] (-1.2,-0.3) -- (-1.2, 0.7);
\draw[thick] (2.6,-0.3) -- (2.7, 0.2);
\draw[thick] (2.7,0.2) -- (2.6, 0.7);

\draw (-0.7, 0) rectangle (-0.3, 0.4);
\draw [fill=gray!30] (0,0) rectangle (0.4,0.4);
\draw [fill=gray!30] (0.4,0) rectangle (0.8,0.4);
\draw (0.8,0) rectangle (1.2,0.4);
\draw (1.2,0) rectangle (1.6,0.4);
\draw (1.6,0) rectangle (2,0.4);
\draw (2,0) rectangle (2.4,0.4);

\node at (-0.95, 0.2) {\texttt{s}};
\draw[black,fill=black] (-0.5,0.2) circle (.4ex);
\node at (1.8,0.2) {$1$};
\draw[black,fill=black] (2.2,0.2) circle (.4ex);
\node at (1,0.2) {$2$};
\draw (1.2,0) -- (1.6,0.4);
\draw[-stealth, thick] (2.2,0.2) to[out=90,in=60,distance=0.4cm] (1,0.42);
\draw[-stealth, thick] (-0.5,0.2) to[out=-90,in=-120,distance=0.4cm] (1.8,-0.02);
\end{tikzpicture}
\captionof{figure}{\AllocatorName{}'s unique superposition representation.}
\label{fig:state-rep-3}
\end{subfigure}
\end{minipage}%
\begin{minipage}[t]{0.46\textwidth}
\vspace*{-3em}
\begin{lstlisting}[label={lst:find-pos}]
fun find_pos[n](l: ptr<list>, x: uint,
                acc: uint) -> uint {
  with {
    let node <- default<list>;
    *l <-> node;
    let (this_x, next) <- node;
    let r <- acc + 1;
    let eq <- this_x == x;
  } do if eq {
    let out <- acc;
  } else {
    let out <- find_pos[n-1](next, x, r);
  }
  return out;
}
\end{lstlisting}
\end{minipage}

\setlength{\abovecaptionskip}{0pt}
\begin{minipage}[t]{0.5\textwidth}
\addtocounter{figure}{-1}
\captionof{figure}{Distinct physical memory representations of the same program state in a qRAM of size 6.}
\end{minipage}%
\hspace{1.5em}\begin{minipage}[t]{0.43\textwidth}
\captionof{figure}{Implementation of \lstinline{find_pos}.} \label{fig:find-pos}
\end{minipage}
\end{figure}

For example,~\Cref{fig:state-rep-1} depicts a program state in which \texttt{s} points to memory location 1, which stores value 1 and a pointer to memory location 3, which stores value 2 and \lstinline{null}. \Cref{fig:state-rep-2} depicts a state in which \texttt{s} points to memory location 5, which stores 1 and a pointer to memory location 1, which stores 2 and \lstinline{null}. These two program states are physically distinct but semantically equivalent, causing quantum interference to fail.

\paragraph{\AllocatorName{} Allocator}

We resolve this issue by developing \AllocatorName{}, a memory allocator that ensures that addresses of all dynamic allocations in the program are history-independent.
Under \AllocatorName{}, the single unique program state in which \texttt{s} denotes $[1,2]$ is~\Cref{fig:state-rep-3}, in which the linked list nodes are stored in a uniform superposition of all possible allocation sites.

\paragraph{Radix Trees}

The sorted list implementation of sets has $O(n)$ operation complexity, which is asymptotically inefficient.
By contrast, \citet{bernstein2013} propose using a radix tree~\citep{morrison}, whose linking structure is in one-to-one correspondence with sets and which supports operations in $O(\log n)$ time.
In this work, we construct an efficient and history-independent implementation of sets by combining the radix tree with the \AllocatorName{} dynamic memory allocator.

\subsection{Recursive Data Structure Operations in \LangName{}}

The last tool needed to implement operations on linked data structures is recursion. We demonstrate recursion in \LangName{} by implementing the function \lstinline{find_pos}, which returns the position of an integer \lstinline{x} in a list \lstinline{l}, assuming it is present in the list. This operation traverses to the end of the list, and a natural implementation performs this traversal recursively.

\paragraph{Recursion Bounds}

Classically, a recursive call pushes a return address onto the program stack and sets the program counter to the entry point of the function.
By contrast, a quantum processor does not possess a program counter or the ability to loop potentially infinitely.
All recursive procedures must instead be unrolled to a fixed, classically known depth, and inlined.

\LangName{} guarantees that fully unrolling a recursive function is always possible using a classical \emph{recursion bound} specifying the number of recursion levels. The \LangName{} type system verifies that in a program, all recursive function calls satisfy their respective recursion bound.

\paragraph{Implementing \texttt{find\_pos}}

In~\Cref{fig:find-pos}, we implement \lstinline{find_pos} in \LangName{} by recursively traversing \lstinline{l} and comparing each element with \lstinline{x}.
The implementation takes an accumulator argument \lstinline{acc}, such that calling the function with zero for \lstinline{acc} returns the position of \lstinline{x}.

On line~\ref{lst:find-pos-1}, the bound \lstinline{[n]} denotes that \lstinline{find_pos} may recursively call itself at most \lstinline{n} times.
Lines~\ref{lst:find-pos-3}--\ref{lst:find-pos-13} feature a \lstinline{with}-\lstinline{do} statement, syntactic sugar that facilitates uncomputation. \LangName{} executes this statement by executing the body of \lstinline{with}, then the \lstinline{do}, then the reverse of the \lstinline{with}. In this program, reversing the \lstinline{with}-block uncomputes \lstinline{node}, \lstinline{this_x}, \lstinline{next}, \lstinline{r}, and \lstinline{eq}.

Line~\ref{lst:find-pos-5} swaps the contents of memory at address \lstinline{l} with an empty node \lstinline{node}, and line~\ref{lst:find-pos-6} extracts from \lstinline{node} the content \lstinline{this_x} and pointer \lstinline{next} to the tail.
Line~\ref{lst:find-pos-8} computes \lstinline{eq}, a flag for whether \lstinline{this_x} is equal to \lstinline{x}. If it is, line~\ref{lst:find-pos-10} returns the current position, and if it is not, line~\ref{lst:find-pos-12} recursively calls \lstinline{find_pos} on the tail of the list, with a recursion bound that has decreased by one.

\section{\CoreLang{}: Reversible and Circuit-Based Semantics} \label{sec:core-semantics}

This section presents \CoreLang{}, an imperative programming language that enables quantum programming with random-access memory. The operational semantics of the language is fully reversible, enabling a program to execute forwards or backwards. A \CoreLang{} program admits \emph{inversion}, which syntactically transforms it into another program with reversed semantics. A program also admits a syntactic translation into a unitary quantum circuit. In~\Cref{sec:main-semantics}, we extend \CoreLang{} with co-recursive data types and recursion to produce \LangName{}.

\subsection{Syntax}

The syntax of \CoreLang{} consists of types, values, expressions, and statements.
\begin{align*}
    \textsf{Type}\ \type \Coloneqq{} & \tUnit \mid \tUInt \mid \tBool \mid \tPair{\type_1}{\type_2} \mid \tPtr{\type} \\
    \textsf{Value}\ \vValue \Coloneqq{} & x \mid \vUnit \mid \ePair{x_1}{x_2} \mid \vNum{n} \mid \vTrue \mid \vFalse \mid \vNull{\type} \mid \vPtr{\type}{p} \hspace{1em} (n \in \textsf{UInt}, p \in \textsf{Addr}) \\
    \textsf{Expression}\ \eExp \Coloneqq{} & \vValue \mid \eProj{1}{x} \mid \eProj{2}{x} \mid \eUnop{uop}{x} \mid \eBinop{x_1}{bop}{x_2} \\
    \textsf{Operator}\ uop \Coloneqq{} & \oNot \mid \oTest \hspace{2em} bop \Coloneqq{} \oAnd \mid \oOr \mid \oAdd \mid \oSub \mid \oMul \\
    \textsf{Statement}\ \sStmt \Coloneqq{} & \sSkip \mid \sSeq{\sStmt_1}{\sStmt_2} \mid \sBind{x}{\eExp} \mid \sUnbind{x}{\eExp} \mid \sSwap{x_1}{x_2} \mid \sMemSwap{x_1}{x_2} \mid \sIf{x}{s}
\end{align*}
\paragraph{Types}
The types \type{} are unit, unsigned integers, Booleans, products, and pointers.

\paragraph{Values}
The language distinguishes syntactic values $\vValue$, which are variables, unit, pairs of variables, numeric, Boolean, and pointer literals. Numeric literals $n$ are elements of a set \textsf{UInt} of unsigned integers in the range $[0, 2^k - 1]$ and pointer literals $p$ are elements of a set \textsf{Addr} of addresses in the range $[1, 2^k - 1]$, where the parameter $k$ is the word size of the system.

\paragraph{Expressions}
Expressions $\eExp$ are values, projections, and unary and binary operators. These operators are standard logical and arithmetic operators on Booleans and integers. The operation $\oTest$ tests whether an integer is equal to zero or a pointer is equal to null.

\paragraph{Statements}
Statements $\sStmt$ include $\sSkip$, which has no effect, and composition $\sSeq{\sStmt_1}{\sStmt_2}$ of statements. The assignment statement $\sBind{x}{\eExp}$ initializes a variable $x$ and evaluates $\eExp$, storing the result in $x$. The un-assignment statement $\sUnbind{x}{\eExp}$ is unique to reversible programming --- it evaluates $\eExp$ in reverse to restore $x$ back to its default state, and discards $x$. The swap statement $\sSwap{x_1}{x_2}$ exchanges the contents of variables $x_1$ and $x_2$. The swap-with-memory statement $\sMemSwap{x_1}{x_2}$ exchanges the contents of memory at address $x_1$ with the contents of $x_2$. Finally, the statement $\sIf{x}{s}$ has the effect of $s$ whenever the condition $x$ is true.

\subsection{Type System}
\begin{figure}
\resizebox{\textwidth}{!}{
\parbox{1.1\textwidth}{
\begin{mathpar}
\inferrule[TV-Var]{\vphantom{\ctx}}{\hastype{\ctx, x : \type}{x}{\type}}

\inferrule[TV-Pair]{\hastype{\ctx}{x_1}{\type_1} \\ \hastype{\ctx}{x_2}{\type_2}}{\hastype{\ctx}{\ePair{x_1}{x_2}}{\tPair{\type_1}{\type_2}}}

\inferrule[TV-Num]{\vphantom{\ctx}}{\hastype{\ctx}{\vNum{n}}{\tUInt}}

\inferrule[TV-Bool]{b \in \{\vTrue, \vFalse\}}{\hastype{\ctx}{b}{\tBool}}

\inferrule[TV-Null]{\vphantom{\ctx}}{\hastype{\ctx}{\vNull{\type}}{\tPtr{\type}}}

\inferrule[TV-Ptr]{\vphantom{\ctx}}{\hastype{\ctx}{\vPtr{\type}{p}}{\tPtr{\type}}}

\inferrule[TE-Proj]{\hastype{\ctx}{x}{\tPair{\type_1}{\type_2}}}{\hastype{\ctx}{\eProj{i}{x}}{\type_i}}

\inferrule[TE-Not]{\hastype{\ctx}{x}{\tBool}}{\hastype{\ctx}{\eUnop{\oNot}{x}}{\tBool}}

\inferrule[TE-Test]{\hastype{\ctx}{x}{\type} \\\\ \type \in \{\tUInt, \tPtr{\type'}\}}{\hastype{\ctx}{\eUnop{\oTest}{x}}{\tBool}}

\inferrule[TE-Lop]{\hastype{\ctx}{x_1}{\tBool} \quad \hastype{\ctx}{x_2}{\tBool} \\\\ bop \in \{\oAnd, \oOr\}}{\hastype{\ctx}{\eBinop{x_1}{bop}{x_2}}{\tBool}}

\inferrule[TE-Aop]{\hastype{\ctx}{x_1}{\tUInt} \quad \hastype{\ctx}{x_2}{\tUInt} \\\\ bop \in \{\oAdd, \oSub, \oMul\}}{\hastype{\ctx}{\eBinop{x_1}{bop}{x_2}}{\tUInt}}
\end{mathpar}
}}
\setlength{\abovecaptionskip}{5pt}
\caption{Selected typing rules in \CoreLang{}. The full definition is presented in~\Cref{sec:full-semantics}.} \label{fig:core-types}
\end{figure}
\begin{figure}
\resizebox{\textwidth}{!}{
\parbox{1.1\textwidth}{
\begin{mathpar}
\inferrule[S-Skip]{\vphantom{\ctx}}{\stmtok{\ctx}{\sSkip}{\ctx}}

\inferrule[S-Seq]{\stmtok{\ctx}{\sStmt_1}{\ctx'} \\ \stmtok{\ctx'}{\sStmt_2}{\ctx''}}{\stmtok{\ctx}{\sSeq{\sStmt_1}{\sStmt_2}}{\ctx''}}

\inferrule[S-Assign]{\hastype{\ctx}{\eExp}{\type} \\ x \notin \ctx}{\stmtok{\ctx}{\sBind{x}{e}}{\ctx, x : \type}}

\inferrule[S-UnAssign]{\hastype{\ctx}{\eExp}{\type} \\ x \notin \ctx}{\stmtok{\ctx, x : \type}{\sUnbind{x}{e}}{\ctx}}

\inferrule[S-Swap]{\hastype{\ctx}{x_1}{\type} \\ \hastype{\ctx}{x_2}{\type}}{\stmtok{\ctx}{\sSwap{x_1}{x_2}}{\ctx}}

\inferrule[S-MemSwap]{\hastype{\ctx}{x_1}{\tPtr{\type}} \\ \hastype{\ctx}{x_2}{\type}}{\stmtok{\ctx}{\sMemSwap{x_1}{x_2}}{\ctx}}

\inferrule[S-If]{\stmtok{\ctx}{\sStmt}{\ctx} \\ \hastype{\ctx}{x}{\tBool} \\ x \notin \modified{\sStmt}}{\stmtok{\ctx}{\sIf{x}{s}}{\ctx}}
\end{mathpar}
}}
\setlength{\abovecaptionskip}{5pt}
\caption{Well-formation rules for statements in \CoreLang{}.} \label{fig:core-stmt-ok}
\end{figure}
The type system of \CoreLang{} assigns a type to a value or expression and determines whether a statement is well-formed and corresponds to a valid quantum program.

\paragraph{Typing}
In~\Cref{fig:core-types}, we define the typing judgments for values and expressions. A context $\ctx$ is a mapping from variables $x$ to types $\type$. The judgment $\hastype{\ctx}{\vValue}{\type}$ states that under context $\ctx$, the value $\vValue$ has type $\type$.
The judgment $\hastype{\ctx}{\eExp}{\type}$ states that under $\ctx$, the expression $\eExp$ has type $\type$.

In~\Cref{fig:core-stmt-ok}, we define the judgment $\stmtok{\ctx}{\sStmt}{\ctx'}$, stating that under $\ctx$, statement $\sStmt$ is well-formed and yields new context $\ctx'$. The statement $\sSkip$ has no effect, and sequencing of $\sStmt_1$ and $\sStmt_2$ composes their effect on the context. Assignment $\sBind{x}{\eExp}$ is well-defined when $\ctx$ does not contain $x$, and produces a context giving $x$ the type of $\eExp$. Un-assignment is the reverse, removing $x$ from the context. A swap requires its arguments to have the same type, and a swap with memory requires the first argument to be a pointer to the type of the second.

An \texttt{if}-statement requires its branch $\sStmt$ to have no net effect on the context. For reversibility, it requires the condition to be a Boolean that is not modified by $\sStmt$, as defined by the judgement:
\begin{alignat*}{5}
\modified{\sSkip} &= \emptyset \quad & \modified{\sSwap{x_1}{x_2}} &= \{x_1, x_2\} \\
\modified{\sSeq{s_1}{s_2}} &= \modified{s_1} \cup \modified{s_2} \quad & \modified{\sMemSwap{x_1}{x_2}} &= \{x_2\} \\
\modified{\sBind{x}{e}} &= \modified{\sUnbind{x}{e}} = \{x\} \quad & \modified{\sIf{x}{s}} &= \modified{s}
\end{alignat*}
\subsection{Reversible Operational Semantics}

The operational semantics of \CoreLang{} describes reversible operations that may be lifted to act in superposition over a quantum state. The semantics specifies how the machine state evolves during program evaluation while preserving sufficient information to enable reversibility.

\paragraph{Machine State}
The semantics transitions over \emph{machine states} $\mstate{\sStmt, \reg, \mem}$. Machine states consist of a program $\sStmt$, a \emph{register file} $\reg$ mapping variables to values, and a \emph{memory} $\mem$ mapping addresses to values. The register file $\reg$ corresponds to the main quantum registers over which we may perform arbitrary gates, while the memory $\mem$ corresponds to the qRAM. We use $\regtypes{\reg}$ to denote the typing context formed by mapping each variable in $\reg$ to the type of the value stored.

\paragraph{Default Value}
The default value of a type $\type$, denoted $\default{\type}$, is the initial value of a variable of type $\type$ and corresponds to a convenient memory representation of constant zeroes of appropriate size:
\begin{alignat*}{5}
\default{\tUnit} &= \vUnit \quad & \default{\tBool} &= \vFalse \quad & \default{\tPtr{\type}} &= \vNull{\type} \\
\default{\tUInt} &= \vNum{0} \quad & \default{\tPair{\type_1}{\type_2}} &= \ePair{\default{\type_1}}{\default{\type_2}} &
\end{alignat*}
The number of bits required to represent a type is the same across a superposition of values of that type, and can be determined statically. When $\type_1$ and $\type_2$ are represented using the same number of bits, their default values $\default{\type_1}$ and $\default{\type_2}$ are physically identical. For example, $\vNum{0}$ and $\vNull{\type}$ are identical, assuming that integers and pointers are represented using the same number of bits.

\paragraph{Forward Evaluation}
In~\Cref{fig:core-exp-eval}, we define the forward expression evaluation judgment \linebreak$\estep{\eExp}{\reg}{x}{z}{\vValue}$. This judgment states that expression $\eExp$ evaluates to value $\vValue$, which is stored in register $x$, whose initial state was $z$. This judgment maintains reversibility by ensuring that 1) $\eExp$ is not destroyed by evaluation, 2) $\vValue$ is stored into a particular location $x$, and 3) the initial state $z$ of variable $x$ is a constant that remains known after this operation.

\begin{figure}
\resizebox{\textwidth}{!}{
\parbox{1.1\textwidth}{
\begin{mathpar}
\inferrule[SE-Var]{\get{\reg}{x'} = \vValue \\ \hastype{\cdot}{\vValue}{\type}}{\estep{x'}{\reg}{x}{\default{\type}}{\vValue}}

\inferrule[SE-Pair]{\get{\reg}{x_1} = \vValue_1 \\ \get{\reg}{x_2} = \vValue_2 \\ \hastype{\cdot}{\vValue_1}{\type_1} \\ \hastype{\cdot}{\vValue_2}{\type_2}}{\estep{\ePair{x_1}{x_2}}{\reg}{x}{\ePair{\default{\type_1}}{\default{\type_2}}}{\ePair{\vValue_1}{\vValue_2}}}

\inferrule[SE-Val]{v \neq x' \\ v \neq (x_1, x_2) \\ \hastype{\cdot}{\vValue}{\type}}{\estep{\vValue}{\reg}{x}{\default{\type}}{\vValue}}

\inferrule[SE-Proj]{\get{\reg}{x'} = \ePair{\vValue_1}{\vValue_2} \\ \hastype{\cdot}{\vValue_i}{\type}}{\estepline{\eProj{i}{x'}}{\reg}{x}{\default{\type}}{\vValue_i}}

\inferrule[SE-Not]{\get{\reg}{x'} = b \\ b \in \{\vTrue, \vFalse\}}{\estepline{\eUnop{\oNot}{x'}}{\reg}{x}{\vFalse}{\neg b}}

\inferrule[SE-Lop]{\get{\reg}{x_1} = b_1 \\ \get{\reg}{x_2} = b_2 \\\\ bop \in \{\oAnd, \oOr\} \\ b_1, b_2 \in \{\vTrue, \vFalse\}}{\estepline{\eBinop{x_1}{bop}{x_2}}{\reg}{x}{\vFalse}{b_1\ bop\ b_2}}

\inferrule[SE-Aop]{\get{\reg}{x_1} = \vNum{n_1} \\ \get{\reg}{x_2} = \vNum{n_2} \\ bop \in \{\oAdd, \oSub, \oMul\}}{\estepline{\eBinop{x_1}{bop}{x_2}}{\reg}{x}{\vNum{0}}{\vNum{n_1\ bop\ n_2}}}
\end{mathpar}
}}
\caption{Selected forward step rules of expressions in \CoreLang{}. The full definition is presented in~\Cref{sec:full-semantics}.} \label{fig:core-exp-eval}
\end{figure}
\begin{figure}
\resizebox{\textwidth}{!}{
\parbox{1.1\textwidth}{
\begin{mathpar}
\inferrule[SS-SeqStep]{\sstep{\sStmt_1}{\reg}{\mem}{\sStmt'_1}{\reg'}{\mem'}}{\sstep{\sSeq{\sStmt_1}{\sStmt_2}}{\reg}{\mem}{\sSeq{\sStmt'_1}{\sStmt_2}}{\reg'}{\mem'}}

\inferrule[SS-Assign]{\hastype{\regtypes{\reg}}{\eExp}{\type} \\ \estep{\eExp}{\reg}{x}{\default{\type}}{\vValue}}{\sstep{\sBind{x}{\eExp}}{\reg}{\mem}{\sSkip}{\subst{\reg}{x}{\vValue}}{\mem}}

\inferrule[SS-UnAssign]{\hastype{\regtypes{\reg}}{\eExp}{\type} \\ \restep{\eExp}{\reg}{x}{\default{\type}}{\vValue}}{\sstep{\sUnbind{x}{\eExp}}{\subst{\reg}{x}{\vValue}}{\mem}{\sSkip}{\reg}{\mem}}

\inferrule[SS-Swap]{\vphantom{\ctx}}{\sstepline{\sSwap{x_1}{x_2}}{\subst{\reg}{x_1, x_2}{\vValue_1, \vValue_2}}{\mem}{\sSkip}{\subst{\reg}{x_1, x_2}{\vValue_2, \vValue_1}}{\mem}}

\inferrule[SS-MemSwapNull]{\get{\reg}{x_1} = \vNull{\type}}{\sstep{\sMemSwap{x_1}{x_2}}{\reg}{\mem}{\sSkip}{\reg}{\mem}}

\inferrule[SS-MemSwapPtr]{\get{\reg}{x_1} = \vPtr{\type}{p}}{\sstepline{\sMemSwap{x_1}{x_2}}{\subst{\reg}{x_2}{\vValue}}{\subst{\mem}{p}{\vValue'}}{\sSkip}{\subst{\reg}{x_2}{\vValue'}}{\subst{\mem}{p}{\vValue}}}
\end{mathpar}
}}
\caption{Selected forward step rules of statements in \CoreLang{}. The full definition is presented in~\Cref{sec:full-semantics}.} \label{fig:core-stmt-step}
\end{figure}

The first evaluation rule specifies that evaluating a variable $x'$ copies its value into a destination variable $x$ that starts in default state.
The following three rules similarly specify copying a known constant $\vValue$, or a pair of values, or one field of a pair, into $x$.
The rules for unary and binary operators update $x$ based on a criterion computed from the arguments, and the rule for arithmetic operators updates $x$ to an arithmetic result over the arguments.

\paragraph{Reverse Evaluation}
In~\Cref{sec:full-semantics}, we define the reverse expression evaluation judgment \linebreak$\restep{\eExp}{\reg}{x}{z}{\vValue}$. This judgment states that given the result $\vValue$ of evaluating $\eExp$, stored in variable $x$, we may reverse the evaluation of $\eExp$ and restore $x$ back to its initial state $z$.
All rules defining this judgment are reverses of the rules defining the forward evaluation.

\paragraph{Forward Step}
In~\Cref{fig:core-stmt-step}, we define the forward statement step judgment $\sstep{\sStmt}{\reg}{\mem}{\sStmt'}{\reg'}{\mem'}$. This judgment states that given machine state $\mstate{\reg, \mem}$, executing statement $\sStmt$ results in new machine state $\mstate{\reg', \mem'}$ and new statement $\sStmt'$ to be executed next.

The first two rules skip the $\sSkip$ statement and execute sequences left to right. The rule $\sBind{x}{e}$ first allocates a variable $x$ in default state, and then forward evaluates $\eExp$, storing the result $\vValue$ into $x$. The rule for un-assignment $\sUnbind{x}{e}$ reverses the evaluation of $\eExp$ to restore the state of $x$ from $\vValue$ back to the default state, and then deallocates $x$. The rule for $\sSwap{x_1}{x_2}$ exchanges the values of $x_1$ and $x_2$ in $\reg$. The rule for swapping with memory at a null pointer has no effect. The rule to swap $x_2$ with memory at a valid pointer $p$ swaps the value of $x_2$ in $\reg$ with the contents of $\mem$ at address $p$. Finally, the two rules for \texttt{if} either transition to the body $\sStmt$ if the condition passes or to $\sSkip$ if it fails.

\begin{figure}
\resizebox{\textwidth}{!}{
\parbox{1.1\textwidth}{
\begin{mathpar}
\inferrule[RS-SeqStep]{\rsstep{\sStmt_2}{\reg}{\mem}{\sStmt'_2}{\reg'}{\mem'}}{\rsstep{\sSeq{\sStmt_1}{\sStmt_2}}{\reg}{\mem}{\sSeq{\sStmt_1}{\sStmt'_2}}{\reg'}{\mem'}}

\inferrule[RS-Assign]{\hastype{\regtypes{\reg}}{\eExp}{\type} \\ \restep{\eExp}{\reg}{x}{\default{\type}}{\vValue}}{\rsstep{\sBind{x}{\eExp}}{\reg}{\mem}{\sSkip}{\subst{\reg}{x}{\vValue}}{\mem}}

\inferrule[RS-UnAssign]{\hastype{\regtypes{\reg}}{\eExp}{\type} \\ \estep{\eExp}{\reg}{x}{\default{\type}}{\vValue}}{\rsstep{\sUnbind{x}{\eExp}}{\subst{\reg}{x}{\vValue}}{\mem}{\sSkip}{\reg}{\mem}}

\inferrule[RS-Swap]{\vphantom{\ctx}}{\rsstepline{\sSwap{x_1}{x_2}}{\subst{\reg}{x_1, x_2}{\vValue_1, \vValue_2}}{\mem}{\sSkip}{\subst{\reg}{x_1, x_2}{\vValue_2, \vValue_1}}{\mem}}
\end{mathpar}
}}
\caption{Selected reverse step rules of statements in \CoreLang{}. The full definition is provided in~\Cref{sec:full-semantics}.} \label{fig:core-stmt-reverse}
\end{figure}

\paragraph{Reverse Step}
In~\Cref{fig:core-stmt-reverse}, we define the reverse statement step judgment $\rsstep{\sStmt}{\reg}{\mem}{\sStmt'}{\reg'}{\mem'}$. This judgment states that given machine state $\mstate{\reg', \mem'}$, executing statement $\sStmt$ in reverse results in new machine state $\mstate{\reg, \mem}$ and new statement $\sStmt'$ to be executed next.

Its rules are the reverse of the forward step. Sequential statements execute from right to left. The roles of $\sBind{x}{e}$ and $\sUnbind{x}{e}$ are flipped, with the latter now evaluating $e$ forward and the former evaluating backward. Because the reverse of a swap operation is itself, the rules for swap statements are the same. Executing an \texttt{if} in reverse requires reversing its body.

\subsection{Soundness}

We now show that \CoreLang{} satisfies key soundness properties in both execution directions:

\begin{definition}[Valid Machine State]
    A machine state $\mstate{\sStmt, \reg, \mem}$ is \emph{valid} under initial context $\ctx$ and final context $\ctx'$, denoted $\ctx \vdash \mstate{\sStmt, \reg, \mem} \dashv \ctx'$, if and only if $\regtypes{\reg} = \ctx$ and $\stmtok{\ctx}{\sStmt}{\ctx'}$.
\end{definition}

\begin{theorem}[Preservation]
    If $\ctx \vdash \mstate{\sStmt, \reg, \mem} \dashv \ctx'$ and $\sstep{\sStmt}{\reg}{\mem}{\sStmt'}{\reg'}{\mem'}$, then $\ctx' \vdash \mstate{\sStmt', \reg', \mem'} \dashv \ctx''$ for some context $\ctx''$.
\end{theorem}

\begin{proof}
    By induction over the forward operational semantics of \CoreLang{}.
\end{proof}

A valid machine state $\mstate{\sStmt, \reg, \mem}$ fails to make forward progress only if it must execute statement $\sUnbind{x}{e}$ where $\eExp$ does not compute the value of $x$, i.e. $\reg = \subst{\reg'}{x}{\vValue}$ and $\restep{\eExp}{\reg'}{x}{\default{\type}}{\vValue}$ fails to hold. We denote this condition $\sstepstuck{\sStmt}{\reg}{\mem}$ and define it in~\Cref{sec:full-semantics}.

\begin{theorem}[Progress]
    Given a machine state $\mstate{\sStmt, \reg, \mem}$ that is valid under some initial and final contexts, either $\sStmt$ is $\sSkip$ or $\sstepstuck{\sStmt}{\reg}{\mem}$ or $\sstep{\sStmt}{\reg}{\mem}{\sStmt'}{\reg'}{\mem'}$ for some $\sStmt', \reg', \mem'$.
\end{theorem}

\begin{proof}
    By induction over the definition of valid machine states.
\end{proof}

\begin{theorem}[Reverse Preservation]
If $\ctx' \vdash \mstate{\sStmt, \reg', \mem'} \dashv \ctx''$ and $\rsstep{\sStmt}{\reg}{\mem}{\sStmt'}{\reg'}{\mem'}$, then $\ctx \vdash \mstate{\sStmt', \reg, \mem} \dashv \ctx'$ for some context $\ctx$.
\end{theorem}

\begin{proof}
By induction over the reverse operational semantics of \CoreLang{}.
\end{proof}

A valid machine state $\mstate{\sStmt, \reg', \mem'}$ fails to make reverse progress only if it must execute statement $\sBind{x}{e}$ where $\eExp$ does not compute the value of $x$, defined analogously to the forward case. We denote this condition $\rsstepstuck{\sStmt}{\reg'}{\mem'}$ and define it in~\Cref{sec:full-semantics}.

\begin{theorem}[Reverse Progress]
Given a state $\mstate{\sStmt, \reg', \mem'}$ that is valid under some initial and final contexts, either $\sStmt$ is $\sSkip$ or $\rsstepstuck{\sStmt}{\reg'}{\mem'}$ or $\rsstep{\sStmt}{\reg}{\mem}{\sStmt'}{\reg'}{\mem'}$ for some $\sStmt', \reg, \mem$.
\end{theorem}

\begin{proof}
By induction over the definition of valid machine states.
\end{proof}

\subsection{Program Reversibility and Inversion} \label{sec:program-inversion}

Let the notation $\mapsto^*$ denote the reflexive transitive closure of $\mapsto$, and similarly $\mapsto^*_R$ for $\mapsto_R$.
We say that a program $\sStmt$ \emph{terminates} in the forward direction when, given a register file $\reg$ and memory $\mem$, we have $\sstepstar{\sStmt}{\reg}{\mem}{\sSkip}{\reg'}{\mem'}$ for some final $\reg'$ and $\mem'$.
We say that $\sStmt$ terminates in the reverse direction when, given $\reg'$ and $\mem'$, we have $\rsstepstar{\sStmt}{\reg}{\mem}{\sSkip}{\reg'}{\mem'}$ for some $\reg$ and $\mem$.

The following theorem states that every \CoreLang{} program is \emph{reversible}, meaning that it terminates in the forward direction if and only if it does so in the reverse direction.
\begin{theorem}[Reversibility]
    $\sstepstar{\sStmt}{\reg}{\mem}{\sSkip}{\reg'}{\mem'}$ iff $\rsstepstar{\sStmt}{\reg}{\mem}{\sSkip}{\reg'}{\mem'}$.
\end{theorem}

\begin{proof}
Follows from forward and reverse progress and preservation.
\end{proof}
The type system does not statically enforce that a program terminates in the forward direction, and it is the responsibility of the program to correctly uncompute temporary values.

\paragraph{Inversion}

Given a program $\sStmt$, its \emph{inversion} $\reverse{\sStmt}$ is a program whose semantics are reversed:
{\small
\begin{alignat*}{5}
\reverse{\sSkip} &= \sSkip \quad & \reverse{\sBind{x}{e}} &= \sUnbind{x}{e} \quad  & \reverse{\sSwap{x_1}{x_2}} &= \sSwap{x_1}{x_2} \\
\reverse{\sSeq{\sStmt_1}{\sStmt_2}} &= \sSeq{\reverse{\sStmt_2}}{\reverse{\sStmt_1}} \quad & \reverse{\sUnbind{x}{e}} &= \sBind{x}{e} \quad & \reverse{\sMemSwap{x_1}{x_2}} &= \sMemSwap{x_1}{x_2} \\
& & & & \reverse{\sIf{x}{\sStmt}} &= \sIf{x}{\reverse{\sStmt}}
\end{alignat*}}
The following theorem states that the forward semantics of a program is equivalent to the reverse semantics of its syntactic inversion.

\begin{theorem}[Invertibility]
    $\sstepstar{\sStmt}{\reg}{\mem}{\sSkip}{\reg'}{\mem'}$ iff $\rsstepstar{\reverse{\sStmt}}{\reg}{\mem}{\sSkip}{\reg'}{\mem'}$.
\end{theorem}

\begin{proof}
    By induction over the definition of inversion.
\end{proof}

\subsection{Quantum Circuit Semantics} \label{sec:circuit-semantics}

We now develop a semantics of \CoreLang{} that translates a program $\sStmt$ to a unitary quantum circuit in order to execute $\sStmt$ on a quantum computer.

\paragraph{Expression Semantics}
Given an expression $\eExp$ whose forward operational semantics is defined as $\estep{\eExp}{\reg}{x}{\default{\type}}{\vValue}$, we lift it to a unitary gate $U_\eExp$ that operates on register file $\reg$ and a new register $x$ and stores $\vValue$ in this register, that is, $U_\eExp\ket{\reg, \default{\type}} = \ket{\reg, \vValue}$. Each possible expression -- production of constant value, variable reference, projection from a pair, integer arithmetic, and Boolean logic -- can be implemented as a unitary gate manipulating $x$:

\begin{itemize}
    \item A constant value $v$ can be reversibly produced by applying the NOT gate on each bit of $x$ whose corresponding bit is set in the bit representation of $v$.
    \item A variable $x'$ may be reversibly copied into $x$ by applying the controlled-NOT gate on each bit in $x$ conditioned on the corresponding bit of $x'$.
    \item Projection is analogous to copying a variable, except only one field of the pair is copied.
    \item Integer arithmetic may be reversibly performed through circuit designs such as~\citet{draper2000addition,cheng2002,islam2009low}.
    \item Boolean logic may be embedded reversibly into NOT and Toffoli (controlled-controlled-NOT) gates, as established by \citet{Fredkin1982ConservativeL}.
\end{itemize}

\paragraph{Statement Semantics}

\newsavebox{\slashwire}
\savebox{\slashwire}{
\begin{quantikz}
    & \qwbundle{}
\end{quantikz}
}
\newcommand{\varwire}{\smash{\hspace{-0.6em}\raisebox{0.3em}{$\oset{\hspace{1.7em}\tiny k}{\usebox{\slashwire}}$}\hspace{-0.1em}}}

\newsavebox{\tripwire}
\savebox{\tripwire}{
\begin{quantikz}
    & \qwbundle[alternate]{}
\end{quantikz}
}
\newcommand{\varswire}{\hspace{-0.6em}\raisebox{0.3em}{\usebox{\tripwire}}\hspace{-0.1em}}

\newsavebox{\skipcircuit}
\savebox{\skipcircuit}{
\begin{quantikz}
    \lstick{$\ket{\reg, \mem}$} & \rstick{$\ket{\reg, \mem}$} \qwbundle[alternate]{}
\end{quantikz}
}
\newsavebox{\assigncircuit}
\savebox{\assigncircuit}{
\begin{quantikz}
    \lstick{$\ket{\reg}$} & \gate[wires=2]{U_{\eExp}}\qwbundle[alternate]{} & \rstick{$\ket{\reg}$} \qwbundle[alternate]{} \\
    \lstick{$\ket{0}^{\otimes k}$} & \qwbundle{\hspace{-1em}k} & \rstick{$\ket{x}$} \qwbundle{k}
\end{quantikz}
}
\newsavebox{\unassigncircuit}
\savebox{\unassigncircuit}{
\begin{quantikz}
    \lstick{$\ket{\reg}$} & \gate[wires=2]{U_{\eExp}^\dagger}\qwbundle[alternate]{} & \rstick{$\ket{\reg}$} \qwbundle[alternate]{} \\
    \lstick{$\ket{x}$} & \qwbundle{\hspace{-1em}k} & \rstick{$\ket{0}^{\otimes k}$} \qwbundle{k}\end{quantikz}
}
\newsavebox{\seqcircuit}
\savebox{\seqcircuit}{
\begin{quantikz}
    \lstick{$\ket{\reg, \mem}$} & \gate{\circuit{\sStmt_1}}\qwbundle[alternate]{} & \gate{\circuit{\sStmt_2}}\qwbundle[alternate]{} & \rstick{$\ket{\reg', \mem'}$} \qwbundle[alternate]{}
\end{quantikz}
}
\newsavebox{\swapcircuit}
\savebox{\swapcircuit}{
\begin{quantikz}
    \lstick{$\ket{x_1}$} & \swap{1}\qwbundle{\hspace{-1em}k} & \rstick{$\ket{x_1}$} \qwbundle{k} \\
    \lstick{$\ket{x_2}$} & \targX{}\qwbundle{\hspace{-1em}k} & \rstick{$\ket{x_2}$} \qwbundle{k}
\end{quantikz}
}
\newsavebox{\memswapcircuit}
\savebox{\memswapcircuit}{
\begin{quantikz}
    \lstick{$\ket{x_1}$} & \gate[wires=3]{\textrm{qRAM}}\qwbundle{\hspace{-1em}k_1} & \rstick{$\ket{x_1}$} \qwbundle{k_1} \\
    \lstick{$\ket{x_2}$} & \qwbundle{\hspace{-1em}k_2} & \rstick{$\ket{x'_2}$} \qwbundle{k_2} \\
    \lstick{$\ket{\mem}$} & \qwbundle[alternate]{} & \rstick{$\ket{\mem'}$} \qwbundle[alternate]{}
\end{quantikz}
}
\newsavebox{\ifcircuit}
\savebox{\ifcircuit}{
\begin{quantikz}
    \lstick{$\ket{x}$} & \ctrl{1} & \rstick{$\ket{x}$} \qw \\
    \lstick{$\ket{\reg, \mem}$} & \gate{\circuit{s}}\qwbundle[alternate]{} & \rstick{$\ket{\reg', \mem'}$} \qwbundle[alternate]{}
\end{quantikz}
}
\begin{figure}
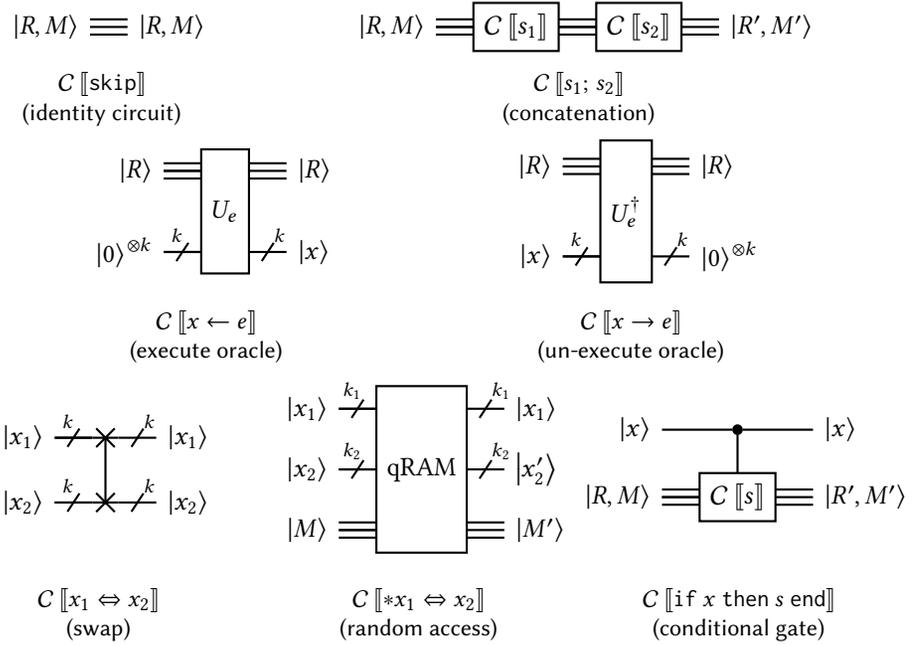

\captionsetup[subfigure]{labelformat=empty}
\captionsetup{justification=centering}
\centering
\begin{subfigure}[b]{.3\textwidth}
\centering
\makebox(100,20){\usebox{\skipcircuit}}
\caption{$\circuit{\sSkip}$ \\ (identity circuit)}
\end{subfigure}
\begin{subfigure}[b]{.6\textwidth}
\centering
\makebox(100,20){\usebox{\seqcircuit}}
\caption{$\circuit{\sSeq{\sStmt_1}{\sStmt_2}}$ \\ (concatenation)}
\end{subfigure}
\begin{subfigure}[b]{.4\textwidth}
\centering
\makebox(100,60){\usebox{\assigncircuit}}
\caption{$\circuit{\sBind{x}{\eExp}}$ \\ (execute oracle)}
\end{subfigure}
\begin{subfigure}[b]{.4\textwidth}
\centering
\makebox(100,60){\usebox{\unassigncircuit}}
\caption{$\circuit{\sUnbind{x}{\eExp}}$ \\ (un-execute oracle)}
\end{subfigure}
\begin{subfigure}[b]{.3\textwidth}
\centering
\makebox(100,75){\usebox{\swapcircuit}}
\caption{$\circuit{\sSwap{x_1}{x_2}}$ \\ (swap)}
\end{subfigure}
\begin{subfigure}[b]{.3\textwidth}
\centering
\makebox(100,75){\usebox{\memswapcircuit}}
\caption{$\circuit{\sMemSwap{x_1}{x_2}}$ \\ (random access)}
\end{subfigure}
\begin{subfigure}[b]{.3\textwidth}
\centering
\makebox(100,75){\usebox{\ifcircuit}}
\caption{$\circuit{\sIf{x}{s}}$ \\ (conditional gate)}
\end{subfigure}
\caption{Definition of quantum circuit semantics of \CoreLang{}.} \label{fig:core-circuit}
\end{figure}

In~\Cref{fig:core-circuit}, we translate a \CoreLang{} statement $\sStmt$ to a circuit fragment $\circuit{\sStmt}$ operating on the register file $\reg$ and memory $\mem$ encoded as quantum states. In the circuits, a wire depicted as \varwire{} denotes a $k$-bit register representing an individual program value. A wire depicted as \varswire{} denotes a collection of such values, such as $\reg$ and $\mem$. A circuit fragment may be expanded to operate on the entire program state by tensor product with the identity gate.

The $\sSkip$ statement translates to the identity circuit that maps a quantum state containing $\reg$ and $\mem$ to itself. Statements $\sStmt_1$ and $\sStmt_2$ are sequenced by concatenating their corresponding circuits $\circuit{\sStmt_1}$ and $\circuit{\sStmt_2}$. An expression $\eExp$ is assigned to $x$ by initializing a zero register, executing $U_\eExp$ on $\ket{R}$ and this register, and naming the resulting register $x$. An expression is un-assigned by executing $U_\eExp^\dagger$, the inverse of $U_\eExp$, producing a register of zeroes that is discarded from the circuit. Variables $x_1$ and $x_2$ are swapped using swap gates across their entire bit register representations. A swap with memory is performed by invoking the quantum random access gate (\Cref{sec:background}) on the address $x_1$, value $x_2$, and memory $\mem$. For an \texttt{if}-statement, its body $\circuit{\sStmt}$ is executed, conditioned on the bit $x$.

\paragraph{Soundness}
The circuit semantics defines a unitary operator $\circuit{\sStmt} : \ket{\reg, \mem} \mapsto \ket{\reg', \mem'}$ for each statement $\sStmt$. It faithfully represents the operational semantics of the statement in \CoreLang{}:
\begin{theorem}[Semantics Equivalence]
    $\circuit{\sStmt}\ket{\reg, \mem} = \ket{\reg', \mem'}$ iff $\sstepstar{\sStmt}{\reg}{\mem}{\sSkip}{\reg'}{\mem'}$
\end{theorem}

\begin{proof}
    By induction on the structure of $\sStmt$.
\end{proof}

\section{\AllocatorName{}: History-Independent Quantum Memory Allocation} \label{sec:memory-allocation}

In this section, we present \AllocatorName{}, a memory allocator for quantum programs. The task of the allocator is to manage dynamic memory --- when invoked by the program, it returns a pointer to an available allocation site and updates internal metadata to mark that site as being in use.

\AllocatorName{} achieves the goals of reversibility, history independence, and bounded-time execution.
Its allocation is reversible, with the reverse being deallocation. It guarantees that allocated data structures are history-independent --- a data structure with a unique logical linking structure has a unique physical memory representation. Finally, it performs allocation in constant time.

To develop \AllocatorName{}, we present in~\Cref{sec:allocator} a core algorithm for allocation based on a \emph{free list} data structure that is reversible and constant-time. Though this algorithm without modification is not history-independent, as we demonstrate in~\Cref{sec:core-interference}, we then present \emph{symmetrization}, \AllocatorName{}'s modification to this algorithm that enables history-independent allocation, in~\Cref{sec:symmetrization}.

\subsection{Core Allocation Algorithm} \label{sec:allocator}

To illustrate the core allocation algorithm, we present an example that stores several values in dynamically allocated memory. We assume that the memory consists of a heap of $m$ words of size $k$, that integers and pointers are word-sized, and that \lstinline{null} and zero have identical bit representations.

\paragraph{Free List}
\AllocatorName{}'s allocation algorithm relies on a \emph{free list} data structure that tracks the free blocks of memory available to the program and is stored in the unallocated portion of the heap.
A pointer to the head of the free list is maintained in a designated \emph{allocation register} denoted \lstinline{alloc_reg}.

In~\Cref{fig:alloc-init}, we depict the initial state of the allocation register and heap.
Before the program executes, \AllocatorName{} constructs this initial state by initializing \lstinline{alloc_reg} to \vNull{} and the blocks of the heap to a sequence $[\vPtr{}{1}$, $\vPtr{}{2}$, $\vPtr{}{3}, \ldots,$ $\vPtr{}{m - 1}$, $\vNull{}]$.\footnotemark{}
Then, it swaps \lstinline{alloc_reg} with the first block of the heap. The effect is that at the start of program execution, the whole heap constitutes the free list, and a pointer to the first free block is stored in \lstinline{alloc_reg}.
\footnotetext{The value of the last block is defined to be \vNull{} so that further allocations return \vNull{} once memory has been exhausted.}
\input{alloc-fig}

\paragraph{Allocation}
In~\Cref{fig:alloc-nums}, we present a program that stores the values 1 and 2 in memory.
Lines~\ref{lst:alloc-nums-1} to~\ref{lst:alloc-nums-3} invoke the \AllocatorName{} allocator to obtain a pointer \lstinline{p1} to space allocated for the first value:
\begin{enumerate}
    \item \lstinline[style=color]{let p1 <- null} initializes \lstinline{p1} to \lstinline{null}, and then
    \item \lstinline[style=color]{p1 <-> alloc_reg} stores in \lstinline{p1} the address of the first available free block and stores zero in \lstinline{alloc_reg}, and then
    \item \lstinline[style=color]{*p1 <-> alloc_reg} stores into \lstinline{alloc_reg} the value at memory address \lstinline{p1}, which is the address of the next available block, and stores zero into memory at address \lstinline{p1}.
\end{enumerate}

In~\Cref{fig:alloc-1}, we depict the program state after executing line~\ref{lst:alloc-nums-3}. The allocation register now points to the new head of the free list, and \lstinline{p1} points to the newly allocated, zero-initialized word.

On lines~\ref{lst:alloc-nums-6} to~\ref{lst:alloc-nums-8}, the program invokes the allocator again to obtain space to store the second value.
In~\Cref{fig:alloc-2}, we depict the state after executing line~\ref{lst:alloc-nums-8}, in which the block to which \lstinline{p1} points has been updated to store 1, and \lstinline{p2} points to another newly allocated and zero-initialized word. Finally, in~\Cref{fig:alloc-3}, we depict the state after line~\ref{lst:alloc-nums-10}, in which \lstinline{p2} now points to the value 2.

By virtue of using the free list, allocation takes constant time. \AllocatorName{} deallocates memory by reversing the allocation operation, which has the effect of putting the newly freed block at the head of the free list. To support variable-sized allocation, \AllocatorName{} partitions the heap into sections of blocks whose sizes are different powers of two, and allocates from the appropriate free list.

\subsection{Core Algorithm Is Not History-Independent} \label{sec:core-interference}

Without modification, the core algorithm does not ensure that data structure representations are history-independent, because the addresses of individual allocations are dependent on the contents of the free list, which vary as the program executes.

In~\Cref{fig:alloc-diff}, we depict how a free list structured differently from~\Cref{fig:alloc-init}, which may possibly be found in the middle of executing a program, would cause the execution of~\Cref{fig:alloc-nums} to produce a final state distinct from~\Cref{fig:alloc-3}. Quantum interference (\Cref{sec:interference}) would fail between~\Cref{fig:alloc-3} and~\Cref{fig:alloc-diff-3} even though these program states are semantically equivalent.

\paragraph{Interference after Deallocation}
As a more extreme and unintuitive example, even after all dynamic memory is deallocated, the representation of the empty heap may cause quantum interference to fail over the remaining register file, leading program outputs to be corrupted.\footnotemark{}%
\footnotetext{An alternative and equivalent mathematical description of the failure is that quantum entanglement~\citep{nielsen_chuang_2010} exists between the memory and the register file, meaning that discarding and measuring the memory would cause the quantum state of the register file to catastrophically collapse from superposition.}%

In~\Cref{fig:entanglement}, we present this example. The initial program state is a superposition:
\[
    \frac{1}{\sqrt{2}}\left(\ket{\pstate{\texttt{l} \goesto [1,2]}} - \ket{\pstate{\texttt{l} \goesto [2,1]}}\right)
\]
The program invokes the operation \lstinline{remove}, which removes and deallocates an element from \lstinline{l}, on the values 1 and 2 in sequence. The result is that all memory is returned to the free list.
Because \lstinline{l} now always points to an empty list, we would expect quantum interference to occur:
\[
    \frac{1}{\sqrt{2}}\left(\ket{\pstate{\texttt{l} \goesto []}} - \ket{\pstate{\texttt{l} \goesto []}}\right) = 0
\]
However, the physical representations of $\pstate{\texttt{l} \goesto []}$, which necessarily incorporate \lstinline{alloc_reg} and the free list stored in memory, differ across the two terms. Thus, quantum interference fails.

As demonstrated by these examples, the core allocation algorithm must be modified to prevent the program from producing an incorrect output after memory allocations and deallocations.

\subsection{Symmetrization Achieves History Independence} \label{sec:symmetrization}

We now develop \emph{symmetrization}, the operation that \AllocatorName{} performs to modify the core allocation algorithm to guarantee history independence.
A key idea is that classically, a \emph{statistical} form of history independence is achieved by making probabilistic choices in the data structure:

\begin{definition}[Statistical History Independence]
    An implementation of an abstract data structure is \emph{statistically history-independent} if and only if, when any two sequences of abstract operations yield equivalent data structures, then the distributions of their physical representations are identical.
\end{definition}

\paragraph{Randomization}
Conceivably, the core allocation algorithm of~\Cref{sec:allocator} may be made statistically history-independent by choosing new allocations uniformly at random. Doing so would require initializing the free list to a uniformly random permutation of the sequence $[\vPtr{}{1}, \ldots, \vPtr{}{m-1}]$.\footnote{Invoking randomness only at initialization means that an allocation immediately following a deallocation will yield the just-deallocated address. This correlation does not violate statistical history independence because the heap always remains in a uniformly random distribution over all possible choices of allocation sites.}

\paragraph{Superposition}
Statistical history independence guarantees a unique probability distribution for the physical representation of the heap. However, this property is insufficient for the support of quantum algorithms, since probability distributions do not exist on equal basis with superpositions in quantum mechanics --- interference only occurs between equal superpositions, not probability distributions, of program states~\citep{bell}. Thus, instead of randomization, \AllocatorName{} uses superposition to achieve full, rather than statistical, history independence for the heap.

\paragraph{Symmetrization}
A single time before the program executes, \AllocatorName{} initializes the free list to a uniform superposition of permutations by invoking a unitary operator that we call \emph{symmetrization}:
\[
    \ket{1, \ldots, n} \mapsto \frac{1}{\sqrt{n!}}\sum_{\pi \in \mathrm{Sym}(n)} \ket{\pi(1), \ldots, \pi(n)}
\]
The notation $\ket{1, \ldots, n}$ denotes a bit string of length $kn$ constituting the concatenation of the $k$-bit binary representations of the integers 1 through $n$.

\begin{theorem}
    If the heap is symmetrized before program execution, then after any sequence of allocation and deallocation operations, the representation of the heap is history-independent.
\end{theorem}

\begin{proof}
    Follows from Theorem 9 in~\citet{teague2001}, replacing equality of probability distributions with equality of superpositions.
\end{proof}

In~\Cref{fig:symmetrization}, we illustrate the result of symmetrization on an empty heap state. Before symmetrization, the free list follows the particular sequence depicted in~\Cref{fig:alloc-init}. Afterward, the free list exists in the unique superposition of all possible permutations, preventing the issue from~\Cref{sec:core-interference}.

In~\Cref{fig:symmetrization-2}, we present the unique physical representation of the program state $\pstate{\texttt{l} \goesto [1,2]}$ under symmetrization, which stores data in a superposition of all possible allocation sites. This uniqueness satisfies history independence and means that quantum interference occurs correctly.

\paragraph{Implementation}

Symmetrization is a central component of quantum simulation, in which it imposes bosonic\footnote{Fermionic antisymmetry, which is closely related, additionally flips the sign according to the parity of the permutation $\pi$.} exchange symmetry~\citep{Liboff1980IntroductoryQM} on the list. A multitude of implementations of this operator exist in the simulation literature, including~\citet{abrams1997} and~\citet{barenco1997}, who achieve implementations in $O(n^2 \log^2 n)$ gates; \citet{chiew2019}, in $O(n^2 \log n)$ gates; and~\citet{Berry2017ImprovedTF}, in $O(n \log n \log \log n)$ gates and parallel time $O(\log n \log \log n)$.\footnote{Comparable to the complexity of the quantum Fourier transform~\citep{qft}.}

\section{\LangName{}: Recursive Data Structure Operations} \label{sec:main-semantics}

This section presents \LangName{}, a language with co-recursive types and bounded recursive functions.

\subsection{Syntax}

The syntax of \LangName{} augments \CoreLang{} with new types, expressions and statements, as well as function declarations $\dDecl$, \emph{recursion bounds} $\bBound$, and programs $\pProgram$. Recursion bounds enable \LangName{} to guarantee that all programs are bounded by classically known parameters.
\begin{alignat*}{3}
    \type \Coloneqq{} & \cdots \mid t \mid \tInd{t}{\type} \quad & \textsf{Declaration}\ \dDecl \Coloneqq{} & \dFun{f}{\bVar?}{x_1 : \type_1, \ldots, x_k : \type_k}{\type}{\sStmt} \\
    \eExp \Coloneqq{} & \cdots \mid \eAlloc{\type} \mid \eCall{f}{\bBound?}{x_1, \ldots, x_k} \quad & \textsf{Bound}\ \bBound \Coloneqq{} & \bVar \mid n \mid \bVar - n' \hspace{1em} (n \in \mathbb{N}, n' \in \mathbb{N} \setminus \{0\}) \\
    \sStmt \Coloneqq{} & \cdots \mid \sReturn{x} \quad & \textsf{Program}\ \pProgram \Coloneqq{} & \dDecl; \pProgram \mid \dFun{\texttt{main}}{}{x : \type}{\type'}{\sStmt}
\end{alignat*}
\paragraph{Types}
The syntax of \LangName{} introduces type variables denoted $t$ as well as co-recursive types $\tInd{t}{\type}$ in which the variable $t$ may appear in $\type$.

\paragraph{Expressions}
The expression $\eAlloc{\type}$ allocates zero-initialized space for a value of type $\type$ and returns a pointer to that type. The expression $\eCall{f}{\bBound?}{x_1, \ldots, x_k}$ invokes function $f$ on arguments $x_1, \ldots x_k$ with recursion bound $\bBound$, or no bound if $\bBound$ is absent.

\paragraph{Statements}
The statement $\sReturn{x}$ returns the value of $x$ to the caller.

\paragraph{Declarations}
The declaration $\dFun{f}{\bVar?}{x_1 : \type_1, \ldots, x_k : \type_k}{\type}{\sStmt}$ declares a function named $f$ with recursion bound variable $\bVar$ (or no bound if absent), arguments $x_1, \ldots, x_k$ with types $\type_1, \ldots, \type_k$ respectively, return type $\type$, and body $\sStmt$. A function body must contain exactly one return statement, as the final statement. \LangName{} permits recursive, but not mutually recursive, functions.\footnote{The presented type system of \LangName{} could be modified straightforwardly to support mutual recursion.}

\paragraph{Bounds}
A recursion bound may be a variable $\bVar$, a classical integer $n$, or their difference.

\paragraph{Programs}
A program is a series of function declarations followed by the entry point \texttt{main}.

\subsection{Type System}

The type system defines well-formed types and enforces that all recursion is classically bounded.

\paragraph{Type Formation}
In~\Cref{sec:full-semantics}, we define the well-formation of types. A context $\tctx$ is a set of type variables currently in scope. The judgment $\typeok{\tctx}{\type}$ states that under context $\tctx$, the type $\type$ is well-formed. All rules are standard except the case for co-recursive types, which additionally requires that a type variable $t$ may only appear under a pointer indirection.

For example, the type of lists of integers is $\texttt{list} \triangleq \tInd{t}{\tPair{\tUInt}{\tPtr{t}}}$, in which we require that the self-reference $t$ must occur under a pointer, similarly to structure definitions in C.
We interpret co-recursive types equi-recursively. For example, we equate the list type $\texttt{list} \triangleq \tInd{t}{\tPair{\tUInt}{\tPtr{t}}}$ with its unfolding $\tPair{\tUInt}{\tPtr{\texttt{list}}}$. In~\Cref{sec:full-semantics}, we define the judgment $\typeequiv{\type_1}{\type_2}$, an equivalence relation that is compatible with unfolding co-recursive types. The default value of a co-recursive type, $\default{\tInd{t}{\type}}$, is defined to be the default value of its unfolding, $\default{\subst{\type}{t}{\tInd{t}{\type}}}$.

\paragraph{Recursive Functions}

The type system of \LangName{} introduces a function context $\fctx$ that maps functions to their recursion bound variable (if present), argument list, and return type, and augments the typing judgments with the function context $\fctx$ and the function $f$ being type-checked.

\begin{figure}
\resizebox{\textwidth}{!}{
\parbox{1.1\textwidth}{
\begin{mathpar}
\inferrule[TE-CallNonRecursive]{f_1 \neq f_2 \\ \get{\fctx}{f_2} = (\textsf{None}, ((\type_1, \ldots, \type_k), \type)) \\\\ \hastype{\ctx}{x_1}{\type_1} \quad \cdots \quad \hastype{\ctx}{x_k}{\type_k}}{\hastypef{\fctx}{\ctx}{f_1}{\eCall{f_2}{}{x_1, \ldots, x_k}}{\type}}

\inferrule[TE-CallBounded]{\bBound = n \text{ or } \bBound = \bVar \text{ or } \bVar - n \text{ where } \textsf{fst}(\get{\fctx}{f_1}) = \bVar \\\\ f_1 \neq f_2 \\ \get{\fctx}{f_2} = (\textsf{Some}\ b_2, ((\type_1, \ldots, \type_k), \type)) \\\\ \hastype{\ctx}{x_1}{\type_1} \quad \cdots \quad \hastype{\ctx}{x_k}{\type_k}}{\hastypef{\fctx}{\ctx}{f_1}{\eCall{f_2}{\bBound}{x_1, \ldots, x_k}}{\type}}

\inferrule[TE-CallSelf]{\get{\fctx}{f} = (\textsf{Some}\ b, ((\type_1, \ldots, \type_k), \type)) \\ \hastype{\ctx}{x_1}{\type_1} \quad \cdots \quad \hastype{\ctx}{x_k}{\type_k}}{\hastypef{\fctx}{\ctx}{f}{\eCall{f}{\bVar - n}{x_1, \ldots, x_k}}{\type}}
\end{mathpar}
}}
\caption{Selected typing rules for expressions in \LangName{}. The full definition is presented in~\Cref{sec:full-semantics}.} \label{fig:main-exp-types}
\end{figure}

In~\Cref{fig:main-exp-types}, we define the new typing judgment for expressions. The judgment $\hastypef{\fctx}{\ctx}{f}{\eExp}{\type}$ states that under function context $\fctx$ and context $\ctx$, inside function $f$, expression $\eExp$ has type $\type$. The type system assigns $\eAlloc{\type}$ the type $\tPtr{\type}$. Most of the remaining rules simply add $\fctx$ and $f$, and we show only the new rules for function calls.

If the function being called, $f_2$, is not recursive, it must be called with no bound and the correct arguments. If it is recursive, it must be called with a bound, which in turn may refer only to the bound variable of the caller function $f_1$. In the last rule, the function $f$ is making a recursive call to itself, which requires the bound to be smaller than the bound variable $\bVar$.

In~\Cref{sec:full-semantics}, we define the new judgment for well-formation of statements. The judgment $\stmtokf{\fctx}{\ctx}{f}{\sStmt}{\ctx'}$ states that under $\fctx$ and $\ctx$, inside function $f$, statement $\sStmt$ is well-formed and yields new context $\ctx'$. Most of the rules simply add $\fctx$ and $f$. The new rule for \texttt{return} requires that the context contains only the returned variable, and yields an empty context.

\subsection{Translation into \CoreLang{}} \label{sec:translation}

\LangName{}'s dynamic semantics is defined by translation into \CoreLang{}. Each statement invoking \texttt{alloc} is replaced with the three \CoreLang{} statements defined in~\Cref{sec:allocator}. Each function call, starting from \texttt{main}, is translated by recursively inlining the function body as follows:

In the general case, a function call has a specified recursion bound.
Given a function call $\sBind{x}{\eCall{f}{n}{x_1, \ldots, x_k}}$, the translator looks up the definition of $f$ and renames its arguments and return value to be $\dFun{f}{\bVar}{x_1 : \type_1, \ldots, x_k : \type_k}{\type}{\sSeq{\sStmt}{\sReturn{x}}}$. It also renames the local variables of $f$ to avoid conflicts with the caller. It then inlines $\sStmt$ into the caller, replacing instances of $\bVar$ with $n$. It recursively expands introduced function calls until the recursion bound reaches zero, and it replaces expressions $\eCall{f}{0}{x_1, \ldots, x_k}$ with $\default{\type}$.
A reverse function call $\sUnbind{x}{\eCall{f}{n}{x_1, \ldots, x_k}}$ is similar, except that the body of $f$ is inverted (\Cref{sec:program-inversion}) before it is inlined.

Additional features of \LangName{}, such as \lstinline{with}-\lstinline{do} blocks, nested expressions, and patterns, also lower into \CoreLang{} and are described in~\Cref{sec:additional-features}.
In~\Cref{sec:translation-example}, we provide a detailed example for the translation of a recursive \LangName{} program into \CoreLang{} by the procedure above.
\begin{theorem}
    In a well-typed program, the process of inlining recursive calls always terminates.
\end{theorem}
\begin{proof}
    Every recursive call has a bound strictly smaller than the bound of the function definition. A call with bound zero terminates immediately. By induction, all other recursive calls terminate.
\end{proof}
As a result, a well-typed \LangName{} program can be converted into a bounded-size quantum circuit.

\subsection{Translation Soundness}

The soundness and reversibility properties of \LangName{} follow from those of \CoreLang{}. A well-typed \LangName{} program translates to a \CoreLang{} program that is well-typed if \emph{coercion} between compatible pointer types is additionally permitted.

Coercion is used by the implementation of $\eAlloc{\type}$ in~\Cref{sec:allocator} to assign a pointer to a free block into a variable of type $\tPtr{\type}$. Since a free block is bitwise compatible with a value of type $\type$, such coercion is always safe with respect to the operational and circuit semantics of \CoreLang{}.

\begin{theorem}
    The inlining of recursive calls and \texttt{alloc} operators in a well-typed \LangName{} program produces a \CoreLang{} program that is well-typed if the coercion of compatible pointers is permitted.
\end{theorem}
\begin{proof}
    The inlining of $\eAlloc{\type}$ is type-correct if coercions are permitted, as discussed above. The type-correctness of the inlining of recursive functions holds by induction on the bound.
\end{proof}

\section{\LibraryName{}: Data Structure Library and Case Study} \label{sec:evaluation}

This section presents \LibraryName{}, a library of abstract data structures we implemented in \LangName{}, as well as a case study of development challenges during its implementation.

\subsection{Language Implementation}

To develop the library, we implemented a \LangName{} interpreter in OCaml that classically executes the operational semantics. The type checker ensures that all recursion is bounded, and the interpreter verifies at runtime that 1) values are uncomputed correctly and 2) all heap memory is freed.

We also implemented a compiler for \LangName{} into quantum circuits. Given a \LangName{} program, the compiler generates for each function a reversible logic circuit that when executed on a bit-representation of its input stores its output into an auxiliary register. The compiler outputs assembly code that invokes primitive logical, arithmetic, and memory operations.

The compiler can further instantiate these primitive operations into multi-controlled-NOT gates. The current implementations of these operations prioritize simplicity; more optimized designs are possible but not yet implemented. The output is in a simple textual netlist format that is amenable to conversion for use in other quantum computation frameworks.

The source code of the \LangName{} interpreter and compiler and the \LibraryName{} library and tests is available in the artifact of this paper~\citep{artifact} and on GitHub.\footnote{\url{https://github.com/psg-mit/tower-oopsla22}}

\subsection{Summary of Data Structure Library} \label{sec:data-structure-library}

\Cref{tbl:results} summarizes the \LibraryName{} data structure library. We implemented the abstract data structures of lists, stacks, queues, strings, and sets.\footnotemark{} In the table, the name of an abstract data structure is followed by its underlying implementation in parentheses. All implementations are history-independent except for hash table-based sets.
\footnotetext{Maps may be implemented as a straightforward augmentation of sets with values stored alongside each key.}

The time complexity for certain pointer-based data structures is a function of $n$, the number of elements in the data structure. For word-based data structures, the complexity of operations instead scales with the system word size $k$,\footnotemark{} and becomes polynomial in $k$ if the operation performs recursion bounded by $k$ and performs operations over words of size $k$ at each level.
\footnotetext{The word size $k$ is in turn logarithmic in the maximum number of elements $n$ of the data structure.}

We report the counts of qubits and gates used by the compiled circuit for each operation, including temporary qubits, after minimal optimization. To focus on the fundamental complexity of each operation, these counts exclude 1) the qubits constituting the qRAM, 2) qubits and gates internal to the implementation of primitive logical, arithmetic, and memory operations, and 3) gates to copy and swap qubits, all of which are highly dependent on parameter and optimization choices. For operations with variable complexity, these counts are functions of $n$ or $k$. Word sizes $4 \le k \le 12$ are assumed for operations whose complexity depends on $k$, and $k = 8$ is assumed otherwise.

In~\Cref{sec:data-structure-descriptions}, we detail the characteristics of each data structure.
During the development process, we encountered several challenges that do not arise in classical programming, and share them with future quantum developers.
In the following sections, we describe three challenges: recursive uncomputation, mutated uncomputation, and branch sequentialization.

\newcommand{\subname}{\textcolor{gray}{$-$}}
\begin{table}
\caption{Overview of \LibraryName{} library. All operations are reversible. ``Recursion''  denotes whether the operation uses bounded recursion. ``Mutation'' indicates whether it leaves input in a mutated state. ``Complexity'' indicates its time complexity in terms of the number of elements $n$ or the word size $k$. ``LoC'' counts the lines of code implementing the operation and its dedicated helpers, excluding calls to other operations listed here. ``Qubits'' counts qubits used, excluding those internal to the qRAM and primitive logic/arithmetic/memory operations. ``Gates'' counts primitive logic/arithmetic/memory gates invoked, excluding copy/swap gates.}
\begin{threeparttable}
\resizebox{\textwidth}{!}{%
\begin{tabular}{ l c c c c c c c }
\toprule
Data Structure                      & Reversible & Recursion & Mutation & Complexity & LoC & Qubits & Gates \\
\midrule
List & & & \\
\subname{} \texttt{length}          & Yes & Yes & No  & $O(n)$ & 20 & $34n + 32$ & $23n + 3$ \\
\subname{} \texttt{sum}             & Yes & Yes & No  & $O(n)$ & 20 & $34n + 40$ & $21n + 3$ \\
\subname{} \texttt{find\_pos}       & Yes & Yes & No  & $O(n)$ & 20 & $42n + 31$ & $19n + 3$ \\
\subname{} \texttt{remove}          & Yes & Yes & Yes & $O(n)$ & 48 & $26n + 56$ & $42n + 3$ \\
Stack (list) & & & \\
\subname{} \texttt{push\_front}     & Yes & No  & Yes & $O(1)$ & 8 & 40 & 4 \\
\subname{} \texttt{pop\_front}      & Yes & No  & Yes & $O(1)$ & 8 & 48 & 4 \\
Queue (list) & & & \\
\subname{} \texttt{push\_back}      & Yes & Yes & Yes & $O(n)$ & 21 & $34n + 32$ & $24n$ \\
\subname{} \texttt{pop\_front}      & Yes & No  & Yes & $O(1)$ & 8 & 48 & 4 \\
String (word) & & & \\
\subname{} \texttt{is\_empty}       & Yes & No  & No  & $O(1)$ & 2 & 25 & 3 \\
\subname{} \texttt{length}          & Yes & No  & No  & $O(1)$ & 2 & 24 & 1 \\
\subname{} \texttt{get\_prefix}     & Yes & No  & No  & $O(k)$ & 8 & $11k$ & 52 \\
\subname{} \texttt{get\_substring}  & Yes & No  & No  & $O(k)$ & 8 & $12k$ & 54 \\
\subname{} \texttt{get}             & Yes & No  & No  & $O(k)$ & 7 & $6k + 1$ & 19 \\
\subname{} \texttt{is\_prefix}      & Yes & Yes & No  & $O(\textrm{poly}(k))$ & 26 & $k^2 + 11k$ & $98k + 3$ \\
\subname{} \texttt{num\_matching}   & Yes & Yes & No  & $O(\textrm{poly}(k))$ & 42 & $k^2 + 13k + 4$ & $110k + 127$ \\
\subname{} \texttt{equal}           & Yes & No  & No  & $O(k)$ & 8 & $6k + 3$ & 5 \\
\subname{} \texttt{concat}          & Yes & No  & No  & $O(k)$ & 9 & $11k$ & 8 \\
\subname{} \texttt{compare}         & Yes & Yes & No  & $O(\textrm{poly}(k))$ & 27 & $5k^2 + 12k$ & $108k + 3$ \\
Set (radix tree) & & & \\
\subname{} \texttt{insert}          & Yes & Yes & Yes & $O(\textrm{poly}(k))$ & 136 & $13k^2 + 21k + 9$ & $1440k^2 + 5056k$ \\
\subname{} \texttt{contains}        & Yes & Yes & No  & $O(\textrm{poly}(k))$ & 334 & $17k^2 + 18k + 2$ & $784k^2 + 1612k + 1$ \\
Set (hash table)${}^\ast$ & & & \\
\subname{} \texttt{insert}          & Yes & Yes & Yes & $O(n)$ & 63 & $52n + 72$ & $68n + 15$ \\
\subname{} \texttt{contains}        & Yes & Yes & No  & $O(n)$ & 7 & $52n + 81$ & $136n + 39$ \\
\bottomrule
\end{tabular}%
}
\begin{tablenotes}
\footnotesize\sffamily
\item${}^\ast$ Hash table-based sets are not history-independent.
\end{tablenotes}
\end{threeparttable}
\label{tbl:results}
\end{table}

\subsection{Challenge 1: Recursive Uncomputation} \label{sec:recursive-uncomputation}

\emph{Recursive uncomputation} arises when the result of a recursive call is uncomputed, and causes certain programs to incur exponential slowdown in asymptotic time complexity. We must be careful to appropriately craft the implementation to avoid naive implementations that cause this slowdown.

\begin{figure}
\centering
\begin{minipage}[t]{.45\textwidth}
\begin{lstlisting}[label={lst:length-exponential}]
fun length[n](l: ptr<list>) -> uint {
  with {
    let l_empty <- l == null;
  } do if l_empty {
    let out <- 0;
  } else with {
    let temp <- default<list>;
    *l <-> temp;
    let next <- temp.2;
    let r <- length[n-1](next);
  } do {
    let out <- r + 1;
  }
  return out;
}
\end{lstlisting}
\end{minipage}%
\begin{minipage}[t]{.55\textwidth}
\begin{lstlisting}[label={lst:length-linear}]
fun length[n](l: ptr<list>, acc: uint) -> uint {
  with {
    let l_empty <- l == null;
  } do if l_empty {
    let out <- acc;
  } else with {
    let temp <- default<list>;
    *l <-> temp;
    let next <- temp.2;
    let r <- acc + 1;
  } do {
    let out <- length[n-1](next, r);
  }
  return out;
}
\end{lstlisting}
\end{minipage}

\setlength{\abovecaptionskip}{0pt}
\begin{minipage}[t]{.45\textwidth}
\caption{$O(2^n)$-time list \lstinline{length}.} \label{fig:length-exponential}
\end{minipage}%
\begin{minipage}[t]{.55\textwidth}
\caption{$O(n)$-time, avoiding recursive uncomputation.} \label{fig:length-linear}
\end{minipage}
\end{figure}
In~\Cref{fig:length-exponential}, we present a correct but inefficient way to compute the length of a linked list.
Line~\ref{lst:length-exponential-3} determines if the list \lstinline{l} is empty. If so, line~\ref{lst:length-exponential-5} sets the output to zero, and if not, lines~\ref{lst:length-exponential-7}--\ref{lst:length-exponential-9} obtain a pointer \lstinline{next} to the tail of the list. Line~\ref{lst:length-exponential-10} recursively computes the length of the tail. Line~\ref{lst:length-exponential-12} adds one to the result of the recursion. After the \lstinline{do}-block completes, \LangName{} reverses the effects of both prior \lstinline{with}-blocks to uncompute all temporary variables and restore the original state of the list.

\paragraph{Recursive Uncomputation}

Though the program in~\Cref{fig:length-exponential} is correct, its time complexity is $O(2^n)$. The reason is that the recursive call on line~\ref{lst:length-exponential-10} produces a temporary value \lstinline{r} that is uncomputed at the end of the \lstinline{do}-block. As a result, \lstinline{length} makes both a forward and a reverse recursive call at each level, resulting in exponential complexity.

Avoiding recursive uncomputation requires not uncomputing the result of recursive calls. One possible approach is inspired by tail recursion from classical programming, where the result of a recursive call is returned without any intervening operations.

In~\Cref{fig:length-linear}, we avoid recursive uncomputation using an extra accumulator argument. Line~\ref{lst:length-linear-5} returns the accumulator in the base case. Line~\ref{lst:length-linear-10} adds one to the accumulator, and line~\ref{lst:length-linear-12} recursively calls the function with the new accumulator and directly assigns the result to the output.

Calling this function with an initial accumulator of zero returns the length. In the new implementation, no second recursive call is necessary, and the implementation achieves $O(n)$ complexity.

We encountered this challenge in the implementations of \lstinline{find_pos}, \lstinline{length}, \lstinline{remove}, and \lstinline{sum} for lists as well as \lstinline{is_prefix}, \lstinline{num_matching}, and \lstinline{compare} for strings.

\subsection{Challenge 2: Mutated Uncomputation} \label{sec:mutated-uncomputation}

\emph{Mutated uncomputation} refers to when uncomputing a temporary value is not straightforward because the program mutates data under conditional branches. The typical strategy of uncomputing a temporary using the expression that originally computed it fails when variables in that expression have been mutated by intervening statements.

In~\Cref{fig:push-back}, we implement the \lstinline{push_back} operation, which adds \lstinline{x} to the end of \lstinline{l}. Line~\ref{lst:push-back-2} computes the flag \lstinline{l_empty}. If it is true, lines~\ref{lst:push-back-4}--\ref{lst:push-back-7} allocate a new node storing \lstinline{x} and set \lstinline{l} to point to it. Otherwise, line~\ref{lst:push-back-12} extracts the \lstinline{tail} of the list, and line~\ref{lst:push-back-15} recursively calls \lstinline{push_back} on it.

\setlength{\intextsep}{0pt}%
\begin{wrapfigure}[23]{r}{.46\textwidth}
\begin{lstlisting}[label={lst:push-back}]
fun push_back[n](l: ptr<list>, x: uint) {
  let l_empty <- l == null;
  if l_empty {
    let head <- alloc<list>;
    l <-> head;
    let head -> null;
    let node <- (x, null);
  } else {
    let node <- default<list>;
    *l <-> node;
    with {
      let (h, tail) <- node;
      let node -> (h, tail);
    } do {
      let () <- push_back[n-1](tail, x);
    }
  }
  let l_empty -> node.2 == null;
  *l <-> node;
  let node -> default<list>;
  return ();
}
\end{lstlisting}
\setlength{\abovecaptionskip}{0pt}
\caption{Implementation of \lstinline{push_back}.} \label{fig:push-back}
\end{wrapfigure}

\paragraph{Mutated Uncomputation}

Line~\ref{lst:push-back-18} uncomputes \lstinline{l_empty} differently from its computation on line~\ref{lst:push-back-2}. The reason is that \lstinline{l} may have been modified by line~\ref{lst:push-back-5}, which would break the relationship between \lstinline{l} and \lstinline{l_empty}. Instead, the program uncomputes \lstinline{l_empty} based on whether the tail of the new list stored in \lstinline{node} is \lstinline{null} on line~\ref{lst:push-back-18}, which is true if and only if \lstinline{l_empty} was true after line~\ref{lst:push-back-2}.

As shown, correctly uncomputing certain values requires reasoning about the invariants that hold on the program state after mutation.
We encountered this challenge in \lstinline{remove} for lists, \lstinline{push_back} for queues, and \lstinline{insert} for sets implemented using both hash tables and radix trees.
By contrast, operations that do not mutate data may typically uncompute temporaries straightforwardly.

\subsection{Challenge 3: Branch Sequentialization} \label{sec:branch-sequentialization}
\begin{figure}
\centering
\vspace*{-0.5em}
\begin{minipage}[t]{.48\textwidth}
\begin{lstlisting}[label={lst:contains-linear}]
fun contains(t: tree, k: uint) -> bool {
  if t == null { let out <- false; }
  else if t.1 == k { let out <- true; }
  else if t.1 < k {
    let out <- contains(t.2, k);
  } else {
    let out <- contains(t.3, k);
  }
  return out;
}
\end{lstlisting}
\end{minipage}%
\begin{minipage}[t]{.52\textwidth}
\begin{lstlisting}[label={lst:contains-log}]
fun contains(t: tree, k: uint) -> bool {
  if t == null { let out <- false; }
  else if t.1 == k { let out <- true; }
  else {
    if t.1 < k { let child <- t.2; }
    else { let child <- t.3; }
    let out <- contains(child, k);
  }
  return out;
}
\end{lstlisting}
\end{minipage}

\setlength{\abovecaptionskip}{0pt}
\begin{minipage}[t]{.45\textwidth}
\caption{$O(n)$-time binary tree traversal.} \label{fig:contains-linear}
\end{minipage}%
\begin{minipage}[t]{.55\textwidth}
\caption{$O(\log n)$-time, avoiding branch sequentialization.} \label{fig:contains-log}
\end{minipage}
\end{figure}
\emph{Branch sequentialization} denotes the fact that the time complexity of a conditionally branching quantum program may be larger than its classical equivalent and the developer may need to structure the program differently to avoid this slowdown.

In~\Cref{fig:contains-linear}, we present a pseudo-\LangName{} program that traverses a binary search tree. Classically, if \lstinline{t} is balanced, then the time complexity of the program is $O(\log n)$.

\paragraph{Branch Sequentialization}

However, if \Cref{fig:contains-linear} is compiled to a quantum circuit (\Cref{sec:circuit-semantics}), each \lstinline{if}-statement becomes a gate that conditions over the gate translation of the \lstinline{if}-body. Each recursive level of the resulting circuit contains a gate that traverses the left child and a gate that traverses the right child, resulting in $O(n)$ total complexity. Put simply, the time complexity of a quantum program is the sum and not the maximum of the complexities of its possible branches.

Efficiently traversing a branching data structure requires making only one recursive call in the function. In~\Cref{fig:contains-log}, we present a program that performs such a traversal in $O(\log n)$.

We encountered this challenge in radix tree \lstinline{insert} and \lstinline{contains}, which use a similar approach to attain poly-logarithmic complexity in the number of elements.

\section{Limitations} \label{sec:future-work}

\paragraph{Error Checking}
Programming with memory causes numerous errors such as null pointer dereference, use after free, and out of memory. Detecting such errors is difficult in quantum programs, which cannot use exceptional control flow. Consequently, \LangName{} defines a dereference of a null pointer to be a no-op and does not dynamically enforce memory safety. Developments in quantum runtime error checking such as~\citet{assertions} may enable detection of memory errors.

\paragraph{Memory Management}
In \LangName{}, the developer must manually reset and deallocate memory, which raises the possibility of automatic memory reclamation. Unfortunately, classical garbage collection traverses the heap, which has a high time cost. Also, erasing information cannot be performed reversibly without uncomputation, which is difficult after memory has become garbage.

\AllocatorName{} incurs internal and external heap fragmentation when allocating variable-sized blocks of memory. Practical dynamic memory management in quantum programs must reconcile space inefficiency of simple allocation schemes with the expense of traversing the heap.

\section{Related Work} \label{sec:related-work}

\paragraph{Quantum Programming with Data Structures}
The three requirements for data structures in superposition were introduced by~\citet{ambainis2003} and elaborated by \citet{bernstein2013} and~\citet{aaronson2019}, who propose using radix trees~\citep{morrison} to implement sets. These authors also define quantum operations such as creating a superposition of all sets.

Other examples of data structures in superposition have been proposed by~\citet{Jeffery_2013,shi2021,chen2021,Booth_2021}.
They may be distinguished from data structures \emph{not} in superposition that are accessed by quantum algorithms for tasks such as simulation~\citep{chakraborty18} and machine learning~\citep{kerenidis17}.

\citet{Pechoux_2020} develop a quantum programming language with inductive data types and recursion. However, in this language it is not possible to construct a superposition over differently linked structures. For example, it is not possible to construct a superposition over linked lists with different lengths, nor over trees with different shapes. Such a restriction is not present in \LangName{}.

In fact, this restriction prevents the implementation of a set that satisfies history independence. To illustrate, suppose we implement sets concretely as sorted, duplicate-free linked lists, as in~\Cref{sec:examples}. Then, if \lstinline{x} denotes a set whose state is a superposition of $[1,2]$ and $[1,3]$, invoking $\texttt{insert}(\texttt{x}, 2)$ should yield a superposition of $[1,2]$ and $[1,2,3]$. However, this superposition is over two lists with different lengths, which is not permitted by~\citet{Pechoux_2020}.

\paragraph{Reversible Programming}

Janus~\citep{Lutz86,tokoyama2007} introduced reversible programming and conditions that are checked after an \texttt{if}-block completes. Consequent developments include \citet{thomsen2015,bowman2011,haulund2017}. Other applications include bidirectional lenses~\citep{bohannon2008} and reversible hardware~\citep{Vieri1995}.

\paragraph{Reversible Memory Management}
Uncomputation~\citep{bennett1989} has become a central feature of quantum programming languages such as~\citet{silq,qsharp,qiskit}. The fact that recursive uncomputation may lead to asymptotic slowdown was noted by~\citet{aaronson2002}. Proposed solutions include whole-program uncomputation~\citep{unqomp} and heuristics trading off between space and time~\citep{square}.

\citet{axelsen2013} present a reversible free list memory allocator, which inspired parts of \AllocatorName{} but is not history-independent.
\citet{buhrman2022} propose a quantum allocator that traverses a prefix tree stored in qRAM for every allocation. By contrast, \AllocatorName{} symmetrizes the heap once and then performs allocation in constant time.

\paragraph{History Independence}

\citet{teague2001} introduced history independence in the context of privacy and cryptography, in which it guarantees that an adversary with access to the physical memory of a system cannot recover sensitive information about the sequence of events that produced the current state. They present history-independent versions of hash tables and a general scheme for memory allocation that inspired parts of \AllocatorName{}. Later work~\citep{hartline,bajaj} elaborates the deterministic and statistical forms of history independence and provides additional schemes for history-independent data structure construction.

By contrast, oblivious data structures~\citep{wang2014} guarantee that the memory access patterns, e.g. the sequence of dereferenced pointers, of semantically equivalent data structure operations are identical.
Unlike history independence, obliviousness is not necessary for a data structure to operate in superposition. Nevertheless, oblivious techniques such as static and garbled circuits~\citep{zahur,yao} and oblivious RAM~\citep{goldreich} enable construction of fixed-control-flow data structures and may be useful in the quantum domain.

\section{Conclusion} \label{sec:conclusion}

Data structures establish abstractions over mathematical objects, such as lists and sets, that are taken for granted in the design of algorithms and ubiquitous in classical programming. However, data structure implementations suitable under quantum superposition must satisfy additional properties, and existing quantum programming frameworks do not support their construction.

In this work, we introduce \LangName{}, a language for quantum programming with data structures, and \AllocatorName{}, a memory allocator that operates in quantum superposition, along with \LibraryName{}, a library of abstract data structures, including the first implementation of the set data structure satisfying the requirements laid out by quantum algorithms. We hope that the library will be useful for quantum developers and that \LangName{} and \AllocatorName{} will lead to even more useful abstractions.

Unlocking the potential of quantum computation rests on our ability to leverage the accumulated wisdom of algorithm and software design while grappling with physical and practical constraints that are unfamiliar to classical developers. As we have shown, understanding classically-motivated principles of computing, such as history independence, allows us to build bridges to abstractions such as data structures that will enable scalable and accessible quantum software.

\begin{acks}
We would like to thank Jesse Michel, Logan Weber, Tian Jin, and anonymous reviewers who provided feedback on drafts of this paper. We thank Scott Aaronson for directing us to the quantum algorithms that were critical to the motivation of this work. We also give special thanks to Chris McNally for invaluable discussions on the nature of quantum random-access memory, and introducing to the authors the physical concept of exchange symmetry.

This work was supported in part by the MIT-IBM Watson AI Lab and the Sloan Foundation. Any opinions, findings, and conclusions or recommendations expressed in this material are those of the authors and do not necessarily reflect the views of the funding agencies.
\end{acks}

\bibliography{biblio.bib}

\newpage

\appendix
\section{Full Semantics} \label{sec:full-semantics}

\begin{figure}[!htb]
\begin{mathpar}
\inferrule[TV-Var]{\vphantom{\ctx}}{\hastype{\ctx, x : \type}{x}{\type}}

\inferrule[TV-Unit]{\vphantom{\ctx}}{\hastype{\ctx}{\vUnit}{\tUnit}}

\inferrule[TV-Pair]{\hastype{\ctx}{x_1}{\type_1} \\ \hastype{\ctx}{x_2}{\type_2}}{\hastype{\ctx}{\ePair{x_1}{x_2}}{\tPair{\type_1}{\type_2}}}

\inferrule[TV-Num]{\vphantom{\ctx}}{\hastype{\ctx}{\vNum{n}}{\tUInt}}

\inferrule[TV-Bool]{b \in \{\vTrue, \vFalse\}}{\hastype{\ctx}{b}{\tBool}}

\inferrule[TV-Null]{\vphantom{\ctx}}{\hastype{\ctx}{\vNull{\type}}{\tPtr{\type}}}

\inferrule[TV-Ptr]{\vphantom{\ctx}}{\hastype{\ctx}{\vPtr{\type}{p}}{\tPtr{\type}}}
\end{mathpar}
\caption{Typing rules for values in \CoreLang{}.} \label{fig:core-val-types-full}
\end{figure}

\begin{figure}[!htb]
\begin{mathpar}
\inferrule[TE-Val]{\hastype{\ctx}{\vValue}{\type}}{\hastype{\ctx}{\vValue}{\type}}

\inferrule[TE-Proj]{\hastype{\ctx}{x}{\tPair{\type_1}{\type_2}}}{\hastype{\ctx}{\eProj{i}{x}}{\type_i}}

\inferrule[TE-Not]{\hastype{\ctx}{x}{\tBool}}{\hastype{\ctx}{\eUnop{\oNot}{x}}{\tBool}}

\inferrule[TE-Test]{\hastype{\ctx}{x}{\type} \\ \type \in \{\tUInt, \tPtr{\type'}\}}{\hastype{\ctx}{\eUnop{\oTest}{x}}{\tBool}}

\inferrule[TE-Lop]{\hastype{\ctx}{x_1}{\tBool} \\ \hastype{\ctx}{x_2}{\tBool} \\ bop \in \{\oAnd, \oOr\}}{\hastype{\ctx}{\eBinop{x_1}{bop}{x_2}}{\tBool}}

\inferrule[TE-Aop]{\hastype{\ctx}{x_1}{\tUInt} \\ \hastype{\ctx}{x_2}{\tUInt} \\ bop \in \{\oAdd, \oSub, \oMul\}}{\hastype{\ctx}{\eBinop{x_1}{bop}{x_2}}{\tUInt}}
\end{mathpar}
\caption{Typing rules for expressions in \CoreLang{}.} \label{fig:core-exp-types-full}
\end{figure}

\begin{figure}[!htb]
\begin{mathpar}
\inferrule[SE-Var]{\get{\reg}{x'} = \vValue \\ \hastype{\cdot}{\vValue}{\type}}{\estep{x'}{\reg}{x}{\default{\type}}{\vValue}}

\inferrule[SE-Pair]{\get{\reg}{x_1} = \vValue_1 \\ \get{\reg}{x_2} = \vValue_2 \\ \hastype{\cdot}{\vValue_1}{\type_1} \\ \hastype{\cdot}{\vValue_2}{\type_2}}{\estep{\ePair{x_1}{x_2}}{\reg}{x}{\ePair{\default{\type_1}}{\default{\type_2}}}{\ePair{\vValue_1}{\vValue_2}}}

\inferrule[SE-Val]{v \text{ is not a variable or pair} \\ \hastype{\cdot}{\vValue}{\type}}{\estep{\vValue}{\reg}{x}{\default{\type}}{\vValue}}

\inferrule[SE-Proj]{\get{\reg}{x'} = \ePair{\vValue_1}{\vValue_2} \\ \hastype{\cdot}{\vValue_i}{\type}}{\estep{\eProj{i}{x'}}{\reg}{x}{\default{\type}}{\vValue_i}}

\inferrule[SE-Not]{\get{\reg}{x'} = b \\ b \in \{\vTrue, \vFalse\}}{\estep{\eUnop{\oNot}{x'}}{\reg}{x}{\vFalse}{\neg b}}

\inferrule[SE-TestNum]{\get{\reg}{x'} = \vNum{n}}{\estep{\eUnop{\oTest}{x'}}{\reg}{x}{\vFalse}{n = 0}}

\inferrule[SE-TestPtr]{\get{\reg}{x'} = \vPtr{\type}{p}}{\estep{\eUnop{\oTest}{x'}}{\reg}{x}{\vFalse}{\vFalse}}

\inferrule[SE-TestNull]{\get{\reg}{x'} = \vNull{\type}}{\estep{\eUnop{\oTest}{x'}}{\reg}{x}{\vFalse}{\vTrue}}

\inferrule[SE-Lop]{\get{\reg}{x_1} = b_1 \\ \get{\reg}{x_2} = b_2 \\ bop \in \{\oAnd, \oOr\} \\ b_1, b_2 \in \{\vTrue, \vFalse\}}{\estep{\eBinop{x_1}{bop}{x_2}}{\reg}{x}{\vFalse}{b_1\ bop\ b_2}}

\inferrule[SE-Aop]{\get{\reg}{x_1} = \vNum{n_1} \\ \get{\reg}{x_2} = \vNum{n_2} \\ bop \in \{\oAdd, \oSub, \oMul\}}{\estep{\eBinop{x_1}{bop}{x_2}}{\reg}{x}{\vNum{0}}{\vNum{n_1\ bop\ n_2}}}
\end{mathpar}
\caption{Forward step rules of expressions in \CoreLang{}.} \label{fig:core-exp-eval-full}
\end{figure}

\begin{figure}[!htb]
\begin{mathpar}
\inferrule[RE-Var]{\get{\reg}{x'} = \vValue \\ \hastype{\cdot}{\vValue}{\type}}{\restep{x'}{\reg}{x}{\default{\type}}{\vValue}}

\inferrule[RE-Pair]{\get{\reg}{x_1} = \vValue_1 \\ \get{\reg}{x_2} = \vValue_2 \\ \hastype{\cdot}{\vValue_1}{\type_1} \\ \hastype{\cdot}{\vValue_2}{\type_2}}{\restep{\ePair{x_1}{x_2}}{\reg}{x}{\ePair{\default{\type_1}}{\default{\type_2}}}{\ePair{\vValue_1}{\vValue_2}}}

\inferrule[RE-Val]{v \text{ is not a variable or pair} \\ \hastype{\cdot}{\vValue}{\type}}{\restep{\vValue}{\reg}{x}{\default{\type}}{\vValue}}

\inferrule[RE-Proj]{\get{\reg}{x'} = \ePair{\vValue_1}{\vValue_2} \\ \hastype{\cdot}{\vValue_i}{\type}}{\restep{\eProj{i}{x'}}{\reg}{x}{\default{\type}}{\vValue_i}}

\inferrule[RE-Not]{\get{\reg}{x'} = b \\ b \in \{\vTrue, \vFalse\}}{\restep{\eUnop{\oNot}{x'}}{\reg}{x}{\vFalse}{\neg b}}

\inferrule[RE-TestNum]{\get{\reg}{x'} = \vNum{n}}{\restep{\eUnop{\oTest}{x'}}{\reg}{x}{\vFalse}{n = 0}}

\inferrule[RE-TestPtr]{\get{\reg}{x'} = \vPtr{\type}{p}}{\restep{\eUnop{\oTest}{x'}}{\reg}{x}{\vFalse}{\vFalse}}

\inferrule[RE-TestNull]{\get{\reg}{x'} = \vNull{\type}}{\restep{\eUnop{\oTest}{x'}}{\reg}{x}{\vFalse}{\vTrue}}

\inferrule[RE-Lop]{\get{\reg}{x_1} = b_1 \\ \get{\reg}{x_2} = b_2 \\ bop \in \{\oAnd, \oOr\} \\ b_1, b_2 \in \{\vTrue, \vFalse\}}{\restep{\eBinop{x_1}{bop}{x_2}}{\reg}{x}{\vFalse}{b_1\ bop\ b_2}}

\inferrule[RE-Aop]{\get{\reg}{x_1} = \vNum{n_1} \\ \get{\reg}{x_2} = \vNum{n_2} \\ bop \in \{\oAdd, \oSub, \oMul\}}{\restep{\eBinop{x_1}{bop}{x_2}}{\reg}{x}{\vNum{0}}{\vNum{n_1\ bop\ n_2}}}
\end{mathpar}
\caption{Reverse step rules of expressions in \CoreLang{}.} \label{fig:core-exp-reverse}
\end{figure}
\begin{figure}[!htb]
\begin{mathpar}
\inferrule[SS-SeqSkip]{\vphantom{\ctx}}{\sstep{\sSeq{\sSkip}{\sStmt}}{\reg}{\mem}{\sStmt}{\reg}{\mem}}

\inferrule[SS-SeqStep]{\sstep{\sStmt_1}{\reg}{\mem}{\sStmt'_1}{\reg'}{\mem'}}{\sstep{\sSeq{\sStmt_1}{\sStmt_2}}{\reg}{\mem}{\sSeq{\sStmt'_1}{\sStmt_2}}{\reg'}{\mem'}}

\inferrule[SS-Assign]{\hastype{\regtypes{\reg}}{\eExp}{\type} \\ \estep{\eExp}{\reg}{x}{\default{\type}}{\vValue}}{\sstep{\sBind{x}{\eExp}}{\reg}{\mem}{\sSkip}{\subst{\reg}{x}{\vValue}}{\mem}}

\inferrule[SS-UnAssign]{\hastype{\regtypes{\reg}}{\eExp}{\type} \\ \restep{\eExp}{\reg}{x}{\default{\type}}{\vValue}}{\sstep{\sUnbind{x}{\eExp}}{\subst{\reg}{x}{\vValue}}{\mem}{\sSkip}{\reg}{\mem}}

\inferrule[SS-Swap]{\vphantom{\ctx}}{\sstep{\sSwap{x_1}{x_2}}{\subst{\reg}{x_1, x_2}{\vValue_1, \vValue_2}}{\mem}{\sSkip}{\subst{\reg}{x_1, x_2}{\vValue_2, \vValue_1}}{\mem}}

\inferrule[SS-MemSwapNull]{\get{\reg}{x_1} = \vNull{\type}}{\sstep{\sMemSwap{x_1}{x_2}}{\reg}{\mem}{\sSkip}{\reg}{\mem}}

\inferrule[SS-MemSwapPtr]{\get{\reg}{x_1} = \vPtr{\type}{p}}{\sstep{\sMemSwap{x_1}{x_2}}{\subst{\reg}{x_2}{\vValue}}{\subst{\mem}{p}{\vValue'}}{\sSkip}{\subst{\reg}{x_2}{\vValue'}}{\subst{\mem}{p}{\vValue}}}

\inferrule[SS-IfTrue]{\get{\reg}{x} = \vTrue}{\sstep{\sIf{x}{\sStmt}}{\reg}{\mem}{\sStmt}{\reg}{\mem}}

\inferrule[SS-IfFalse]{\get{\reg}{x} = \vFalse}{\sstep{\sIf{x}{\sStmt}}{\reg}{\mem}{\sSkip}{\reg}{\mem}}
\end{mathpar}
\caption{Forward step rules of statements in \CoreLang{}.} \label{fig:core-stmt-step-full}
\end{figure}
\begin{figure}[!htb]
\begin{mathpar}
\inferrule[RS-SeqSkip]{\vphantom{\ctx}}{\rsstep{\sSeq{\sStmt}{\sSkip}}{\reg}{\mem}{\sStmt}{\reg}{\mem}}

\inferrule[RS-SeqStep]{\rsstep{\sStmt_2}{\reg}{\mem}{\sStmt'_2}{\reg'}{\mem'}}{\rsstep{\sSeq{\sStmt_1}{\sStmt_2}}{\reg}{\mem}{\sSeq{\sStmt_1}{\sStmt'_2}}{\reg'}{\mem'}}

\inferrule[RS-Assign]{\hastype{\regtypes{\reg}}{\eExp}{\type} \\ \restep{\eExp}{\reg}{x}{\default{\type}}{\vValue}}{\rsstep{\sBind{x}{\eExp}}{\reg}{\mem}{\sSkip}{\subst{\reg}{x}{\vValue}}{\mem}}

\inferrule[RS-UnAssign]{\hastype{\regtypes{\reg}}{\eExp}{\type} \\ \estep{\eExp}{\reg}{x}{\default{\type}}{\vValue}}{\rsstep{\sUnbind{x}{\eExp}}{\subst{\reg}{x}{\vValue}}{\mem}{\sSkip}{\reg}{\mem}}

\inferrule[RS-Swap]{\vphantom{\ctx}}{\rsstep{\sSwap{x_1}{x_2}}{\subst{\reg}{x_1, x_2}{\vValue_1, \vValue_2}}{\mem}{\sSkip}{\subst{\reg}{x_1, x_2}{\vValue_2, \vValue_1}}{\mem}}

\inferrule[RS-MemSwapNull]{\get{\reg}{x_1} = \vNull{\type}}{\rsstep{\sMemSwap{x_1}{x_2}}{\reg}{\mem}{\sSkip}{\reg}{\mem}}

\inferrule[RS-MemSwapPtr]{\get{\reg}{x_1} = \vPtr{\type}{p}}{\rsstep{\sMemSwap{x_1}{x_2}}{\subst{\reg}{x_2}{\vValue}}{\subst{\mem}{p}{\vValue'}}{\sSkip}{\subst{\reg}{x_2}{\vValue'}}{\subst{\mem}{p}{\vValue}}}

\inferrule[RS-IfTrue]{\get{\reg}{x} = \vTrue}{\rsstep{\sIf{x}{\sStmt}}{\reg}{\mem}{\sStmt}{\reg}{\mem}}

\inferrule[RS-IfFalse]{\get{\reg}{x} = \vFalse}{\rsstep{\sIf{x}{\sStmt}}{\reg}{\mem}{\sSkip}{\reg}{\mem}}
\end{mathpar}
\caption{Reverse step rules of statements in \CoreLang{}.} \label{fig:core-stmt-reverse-full}
\end{figure}

\begin{figure}[!htb]
\begin{mathpar}
\inferrule[Error-SeqStep]{\sstepstuck{\sStmt_1}{\reg}{\mem}}{\sstepstuck{\sSeq{\sStmt_1}{\sStmt_2}}{\reg}{\mem}}

\inferrule[Error-UnAssign]{\hastype{\regtypes{\reg}}{\eExp}{\type} \\ \lnot \restep{\eExp}{\reg}{x}{\default{\type}}{\vValue}}{\sstepstuck{\sUnbind{x}{\eExp}}{\subst{\reg}{x}{\vValue}}{\mem}} \\

\inferrule[RError-SeqStep]{\rsstepstuck{\sStmt_2}{\reg}{\mem}}{\rsstepstuck{\sSeq{\sStmt_1}{\sStmt_2}}{\reg}{\mem}}

\inferrule[RError-Assign]{\hastype{\regtypes{\reg}}{\eExp}{\type} \\ \lnot \estep{\eExp}{\reg}{x}{\default{\type}}{\vValue}}{\rsstepstuck{\sBind{x}{\eExp}}{\subst{\reg}{x}{\vValue}}{\mem}}
\end{mathpar}
\caption{Errors in \CoreLang{}.}
\end{figure}

\begin{figure}[!htb]
\begin{mathpar}
\inferrule[TypOk-Var]{\vphantom{\ctx}}{\typeok{\tctx, t}{t}}

\inferrule[TypOk-Base]{\type \in \{\tUnit, \tUInt, \tBool\}}{\typeok{\tctx}{\type}}

\inferrule[TypOk-Pair]{\typeok{\tctx}{\type_1} \quad \typeok{\tctx}{\type_2}}{\typeok{\tctx}{\tPair{\type_1}{\type_2}}}

\inferrule[TypOk-Ptr]{\typeok{\tctx}{\type}}{\typeok{\tctx}{\tPtr{\type}}}

\inferrule[TypOk-Ind]{\typeok{\tctx, t}{\type} \\\\ t \notin \exposed{\type}}{\typeok{\tctx}{\tInd{t}{\type}}}
\end{mathpar}
\begin{align*}
& \exposed{\tUnit} = \exposed{\tUInt} = \exposed{\tBool} = \exposed{\tPtr{\type}} = \emptyset \\
& \exposed{t} = \{t\} \quad \exposed{\tInd{t}{\type}} = \exposed{\type} \setminus \{t\} \quad \exposed{\tPair{\type_1}{\type_2}} = \exposed{\type_1} \cup \exposed{\type_2}
\end{align*}
\caption{Well-formation rules for types in \LangName{}.} \label{fig:main-type-ok}
\end{figure}

\begin{figure}[!htb]
\begin{mathpar}
\inferrule[TypEq-Refl]{\type \in \{t, \tUnit, \tUInt, \tBool\}}{\typeequiv{\type}{\type}}

\inferrule[TypEq-Pair]{\typeequiv{\type_1}{\type'_1} \\ \typeequiv{\type_2}{\type'_2}}{\typeequiv{\tPair{\type_1}{\type_2}}{\tPair{\type'_1}{\type'_2}}}

\inferrule[TypEq-Ptr]{\typeequiv{\type}{\type'}}{\typeequiv{\tPtr{\type}}{\tPtr{\type'}}}

\inferrule[TypEq-Ind]{\typeequiv{\type}{\type'}}{\typeequiv{\tInd{t}{\type}}{\tInd{t}{\type'}}}

\inferrule[TypEq-Unfold]{\vphantom{\ctx}}{\typeequiv{\tInd{t}{\type}}{\subst{\type}{t}{\tInd{t}{\type}}}}

\inferrule[TypEq-Sym]{\typeequiv{\type_1}{\type_2}}{\typeequiv{\type_2}{\type_1}}

\inferrule[TypEq-Trans]{\typeequiv{\type_1}{\type_2} \\ \typeequiv{\type_2}{\type_3}}{\typeequiv{\type_1}{\type_3}}
\end{mathpar}
\caption{Type equivalence rules in \LangName{}.} \label{fig:main-type-equiv}
\end{figure}

\begin{figure}[!htb]
\begin{mathpar}
\inferrule[TV-Eq]{\typeequiv{\type_1}{\type_2} \\ \hastype{\ctx}{\vValue}{\type_1}}{\hastype{\ctx}{\vValue}{\type_2}}

\inferrule[TV-Var]{\vphantom{\ctx}}{\hastype{\ctx, x : \type}{x}{\type}}

\inferrule[TV-Unit]{\vphantom{\ctx}}{\hastype{\ctx}{\vUnit}{\tUnit}}

\inferrule[TV-Pair]{\hastype{\ctx}{x_1}{\type_1} \\ \hastype{\ctx}{x_2}{\type_2}}{\hastype{\ctx}{\ePair{x_1}{x_2}}{\tPair{\type_1}{\type_2}}}

\inferrule[TV-Num]{\vphantom{\ctx}}{\hastype{\ctx}{\vNum{n}}{\tUInt}}

\inferrule[TV-Bool]{b \in \{\vTrue, \vFalse\}}{\hastype{\ctx}{b}{\tBool}}

\inferrule[TV-Null]{\vphantom{\ctx}}{\hastype{\ctx}{\vNull{\type}}{\tPtr{\type}}}

\inferrule[TV-Ptr]{\vphantom{\ctx}}{\hastype{\ctx}{\vPtr{\type}{p}}{\tPtr{\type}}}
\end{mathpar}
\caption{Typing rules for values in \LangName{}.} \label{fig:main-val-types-full}
\end{figure}

\begin{figure}
\begin{mathpar}
\inferrule[TE-Val]{\hastype{\ctx}{\vValue}{\type}}{\hastypef{\fctx}{\ctx}{f}{\vValue}{\type}}

\inferrule[TE-Proj]{\hastype{\ctx}{x}{\tPair{\type_1}{\type_2}}}{\hastypef{\fctx}{\ctx}{f}{\eProj{i}{x}}{\type_i}}

\inferrule[TE-Not]{\hastype{\ctx}{x}{\tBool}}{\hastypef{\fctx}{\ctx}{f}{\eUnop{\oNot}{x}}{\tBool}}

\inferrule[TE-Test]{\hastype{\ctx}{x}{\type} \\ \type \in \{\tUInt, \tPtr{\type'}\}}{\hastypef{\fctx}{\ctx}{f}{\eUnop{\oTest}{x}}{\tBool}}

\inferrule[TE-Lop]{\hastype{\ctx}{x_1}{\tBool} \\ \hastype{\ctx}{x_2}{\tBool} \\ bop \in \{\oAnd, \oOr\}}{\hastypef{\fctx}{\ctx}{f}{\eBinop{x_1}{bop}{x_2}}{\tBool}}

\inferrule[TE-Aop]{\hastype{\ctx}{x_1}{\tUInt} \\ \hastype{\ctx}{x_2}{\tUInt} \\ bop \in \{\oAdd, \oSub, \oMul\}}{\hastypef{\fctx}{\ctx}{f}{\eBinop{x_1}{bop}{x_2}}{\tUInt}}

\inferrule[TE-Alloc]{\vphantom{\ctx}}{\hastypef{\fctx}{\ctx}{f}{\eAlloc{\type}}{\tPtr{\type}}}

\inferrule[TE-CallNonRecursive]{f_1 \neq f_2 \\ \get{\fctx}{f_2} = (\textsf{None}, ((\type_1, \ldots, \type_k), \type)) \\\\ \hastype{\ctx}{x_1}{\type_1} \quad \cdots \quad \hastype{\ctx}{x_k}{\type_k}}{\hastypef{\fctx}{\ctx}{f_1}{\eCall{f_2}{}{x_1, \ldots, x_k}}{\type}}

\inferrule[TE-CallBounded]{\bBound = n \text{ or } \bBound = \bVar \text{ or } \bVar - n \text{ where } \textsf{fst}(\get{\fctx}{f_1}) = \bVar \\\\ f_1 \neq f_2 \\ \get{\fctx}{f_2} = (\textsf{Some}\ b_2, ((\type_1, \ldots, \type_k), \type)) \\\\ \hastype{\ctx}{x_1}{\type_1} \quad \cdots \quad \hastype{\ctx}{x_k}{\type_k}}{\hastypef{\fctx}{\ctx}{f_1}{\eCall{f_2}{\bBound}{x_1, \ldots, x_k}}{\type}}

\inferrule[TE-CallSelf]{\get{\fctx}{f} = (\textsf{Some}\ b, ((\type_1, \ldots, \type_k), \type)) \\\\ \hastype{\ctx}{x_1}{\type_1} \quad \cdots \quad \hastype{\ctx}{x_k}{\type_k}}{\hastypef{\fctx}{\ctx}{f}{\eCall{f}{\bVar - n}{x_1, \ldots, x_k}}{\type}}
\end{mathpar}
\caption{Typing rules for expressions in \LangName{}.} \label{fig:main-exp-types-full}
\end{figure}

\begin{figure}[!htb]
\begin{mathpar}
\inferrule[S-Skip]{\vphantom{\ctx}}{\stmtokf{\fctx}{\ctx}{f}{\sSkip}{\ctx}}

\inferrule[S-Seq]{\stmtokf{\fctx}{\ctx}{f}{\sStmt_1}{\ctx'} \\ \stmtokf{\fctx}{\ctx'}{f}{\sStmt_2}{\ctx''}}{\stmtokf{\fctx}{\ctx}{f}{\sSeq{\sStmt_1}{\sStmt_2}}{\ctx''}}

\inferrule[S-Assign]{\hastypef{\fctx}{\ctx}{f}{\eExp}{\type} \\ x \notin \ctx}{\stmtokf{\fctx}{\ctx}{f}{\sBind{x}{e}}{\ctx, x : \type}}

\inferrule[S-UnAssign]{\hastypef{\fctx}{\ctx}{f}{\eExp}{\type} \\ x \notin \ctx}{\stmtokf{\fctx}{\ctx, x : \type}{f}{\sUnbind{x}{e}}{\ctx}}

\inferrule[S-Swap]{\hastype{\ctx}{x_1}{\type} \\ \hastype{\ctx}{x_2}{\type}}{\stmtokf{\fctx}{\ctx}{f}{\sSwap{x_1}{x_2}}{\ctx}}

\inferrule[S-MemSwap]{\hastype{\ctx}{x_1}{\tPtr{\type}} \\ \hastype{\ctx}{x_2}{\type}}{\stmtokf{\fctx}{\ctx}{f}{\sMemSwap{x_1}{x_2}}{\ctx}}

\inferrule[S-If]{\stmtokf{\fctx}{\ctx}{f}{\sStmt}{\ctx} \\ \hastype{\ctx}{x}{\tBool} \\ x \notin \modified{\sStmt}}{\stmtokf{\fctx}{\ctx}{f}{\sIf{x}{s}}{\ctx}}

\inferrule[S-Return]{\vphantom{\ctx}}{\stmtokf{\fctx}{x : \type}{f}{\sReturn{x}}{\cdot}}
\end{mathpar}
\caption{Well-formation rules for statements in \LangName{}.} \label{fig:main-stmt-ok}
\end{figure}

\FloatBarrier

\section{Additional Features} \label{sec:additional-features}

Apart from minor surface syntax enhancements, the following features of \LangName{} are omitted from the formal grammar, and we use them to improve readability of the example programs in this paper:
\begin{itemize}
\item \emph{\lstinline{if-else} statements}. We lower this construct to single-branch \lstinline{if} statements in the core.
\item \emph{\lstinline{with-do} blocks}. The statement $\sWith{\sStmt_1}{\sStmt_2}$ lowers to $\sSeq{\sSeq{\sStmt_1}{\sStmt_2}}{\reverse{\sStmt_1}}$. That is, the statement in the \lstinline{with}-block is reversed after executing the statement in the \lstinline{do}-block. We provide this construct to easily uncompute temporary values. For example, to store $(\texttt{x1} + \texttt{x2}) + \texttt{x3}$ into \texttt{result}, instead of writing
\begin{lstlisting}[style=nonum]
let temp <- x1 + x2;
let result <- temp + x3;
let temp -> x1 + x2; /* Uncompute temp */
\end{lstlisting}
one may more concisely write
\begin{lstlisting}[style=nonum]
with { let temp <- x1 + x2; }
do { let result <- temp + x3; }
\end{lstlisting}
\item \emph{Nested expressions}. In fact, \LangName{} supports writing $(\texttt{x1} + \texttt{x2}) + \texttt{x3}$ even more concisely through nesting expressions. One may simply write:
\begin{lstlisting}[style=nonum]
let result <- x1 + x2 + x3;
\end{lstlisting}
The lowering of this construct uncomputes all temporary values generated by evaluating nested expressions after their use. Nesting is supported for all expressions other than function arguments and return values as functions may cause side effects.
\item \emph{Patterns}. \LangName{} supports assigning to and un-assigning from patterns, such as tuples and literals, in addition to variables. For example, to project from a tuple, instead of writing
\begin{lstlisting}[style=nonum]
let x <- t.1;
let y <- x.1;
let y -> 0;
let z <- x.2;
let x -> (y, z); /* uncompute x */
let w <- t.2;
\end{lstlisting}
one may more concisely write
\begin{lstlisting}[style=nonum]
let ((0, z), w) <- t;
\end{lstlisting}
\item \emph{Built-in operations}. \LangName{} supports additional unary and binary operators, such as comparisons, bitwise operators, and pointer arithmetic.
\item \emph{Default values}. The operator \lstinline{default<t>} constructs $\default{\type}$, the default value for type \lstinline{t}.
\end{itemize}

\section{Detailed Example of \LangName{} Program Translation} \label{sec:translation-example}

In this section, we provide a detailed example of the translation of a \LangName{} program into \CoreLang{} by the procedure described in~\Cref{sec:translation}.

\begin{figure}
\begin{subfigure}[t]{0.32\textwidth}
\begin{lstlisting}[style=color, xleftmargin=0.5em]
fun f[n]() -> uint {
  with { let y <- f[n-1](); }
  do { let out <- y + 1; }
}
fun main() -> uint {
  let p <- alloc<uint>;
  let p -> alloc<uint>;
  let out <- f[2]();
}
\end{lstlisting}
\captionof{figure}{Initial \LangName{} program.} \label{fig:trans-1}
\end{subfigure}
\begin{subfigure}[t]{0.32\textwidth}
\begin{lstlisting}[style=color, xleftmargin=1.5em]
fun f[n]() -> uint {
  with { let y <- f[n-1](); }
  do { let out <- y + 1; }
}
fun main() -> uint {
  /* let p -> alloc<uint>; */
  let p <- null;
  p <-> alloc_reg;
  *p <-> alloc_reg;
  /* let p -> alloc<uint>; */
  *p <-> alloc_reg;
  p <-> alloc_reg;
  let p <- null;
  let out <- f[2]();
}
\end{lstlisting}
\captionof{figure}{After translating \texttt{alloc}.} \label{fig:trans-2}
\end{subfigure}
\begin{subfigure}[t]{0.32\textwidth}
\begin{lstlisting}[style=color, xleftmargin=2.5em]
fun f[n]() -> uint {
  /* with {...} do {...} */
  let y <- f[n-1]();
  let out <- y + 1;
  let y -> f[n-1]();
}
fun main() -> uint {
  let p <- null;
  p <-> alloc_reg;
  *p <-> alloc_reg;
  *p <-> alloc_reg;
  p <-> alloc_reg;
  let p <- null;
  let out <- f[2]();
}
\end{lstlisting}
\captionof{figure}{After translating \texttt{with}-\texttt{do} block.} \label{fig:trans-3}
\end{subfigure}
\begin{subfigure}[t]{0.32\textwidth}
\begin{lstlisting}[style=color, xleftmargin=0.5em]
fun f[n]() -> uint {
  let y <- f[n-1]();
  let out <- y + 1;
  let y -> f[n-1]();
}
fun main() -> uint {
  let p <- null;
  p <-> alloc_reg;
  *p <-> alloc_reg;
  *p <-> alloc_reg;
  p <-> alloc_reg;
  let p <- null;
  /* let out <- f[2](); */
  let y <- f[1]();
  let out <- y + 1;
  let y -> f[1]();
}
\end{lstlisting}
\captionof{figure}{After inlining \texttt{f[2]()}.} \label{fig:trans-4}
\end{subfigure}
\begin{subfigure}[t]{0.32\textwidth}
\begin{lstlisting}[style=color, xleftmargin=1.5em]
fun f[n]() -> uint {
  let y <- f[n-1]();
  let out <- y + 1;
  let y -> f[n-1]();
}
fun main() -> uint {
  let p <- null;
  p <-> alloc_reg;
  *p <-> alloc_reg;
  *p <-> alloc_reg;
  p <-> alloc_reg;
  let p <- null;
  /* let y <- f[1](); */
  let y2 <- f[0]();
  let y <- y2 + 1;
  let y2 -> f[0]();
  let out <- y + 1;
  /* let y -> f[1](); */
  let y2 <- f[0]();
  let y -> y2 + 1;
  let y2 -> f[0]();
}
\end{lstlisting}
\captionof{figure}{After inlining \texttt{f[1]()}.} \label{fig:trans-5}
\end{subfigure}
\begin{subfigure}[t]{0.32\textwidth}
\begin{lstlisting}[style=color, xleftmargin=2.5em]
fun f[n]() -> uint {
  let y <- f[n-1]();
  let out <- y + 1;
  let y -> f[n-1]();
}
fun main() -> uint {
  let p <- null;
  p <-> alloc_reg;
  *p <-> alloc_reg;
  *p <-> alloc_reg;
  p <-> alloc_reg;
  let p <- null;
  /* let y2 <- f[0](); */
  let y2 <- 0;
  let y <- y2 + 1;
  /* let y2 -> f[0](); */
  let y2 -> 0;
  let out <- y + 1;
  /* let y2 <- f[0](); */
  let y2 <- 0;
  let y -> y2 + 1;
  /* let y2 -> f[0](); */
  let y2 -> 0;
}
\end{lstlisting}
\captionof{figure}{After inlining \texttt{f[0]()}.} \label{fig:trans-6}
\end{subfigure}
\caption{Translation of \LangName{} program into \CoreLang{}.} \label{fig:trans-fig}
\end{figure}

In~\Cref{fig:trans-1}, we present the example program. In the declaration of the recursive function \texttt{f}, the recursion bound is \texttt{n}. The function \texttt{f} invokes itself recursively with the recursion bound \texttt{n-1}, which is strictly smaller than \texttt{n}, and is thus well-typed. The non-recursive function \texttt{main} is the entry point of the program and invokes the function \texttt{f} with the recursion bound 2, which is a constant, and is thus also well-typed. The intended output of this \texttt{main} function is the value 2.

This \LangName{} program is lowered to a \CoreLang{} program by the procedure described in~\Cref{sec:translation}. The procedure starts at \texttt{main}. In~\Cref{fig:trans-2}, it first translates the \texttt{alloc} construct, replacing it with the operations defined in~\Cref{sec:allocator}.
Next, in~\Cref{fig:trans-3} it translates the \texttt{with}-\texttt{do} block as described by~\Cref{sec:additional-features}.
In~\Cref{fig:trans-4}, it inlines the call to function \texttt{f} into \texttt{main}, substituting the bound 2 for \texttt{n}.
In~\Cref{fig:trans-5}, it repeats the inlining of \texttt{f} again, taking care to rename the local variable \texttt{y} to avoid collision.
Finally, in~\Cref{fig:trans-6} it replaces the call to \texttt{f} with recursion bound 0 with the base case of 0, as described in~\Cref{sec:translation}.
The result is a fully unfolded, well-typed \CoreLang{} program. Executing this program yields the output 2 from \texttt{main}.

\section{Data Structure Descriptions} \label{sec:data-structure-descriptions}

In this section, we describe in detail the data structures we implemented and notable differences from their classical equivalents in usage or development process.

\paragraph{Linked List}

We implemented four operations that recursively traverse linked lists:
\begin{itemize}
    \item $\texttt{length}(\texttt{l})$ returns the length of \texttt{l} in $O(n)$ time.
    \item $\texttt{sum}(\texttt{l})$ returns the sum of elements in \texttt{l} in $O(n)$ time.
    \item $\texttt{find\_pos}(\texttt{l}, \texttt{x})$ returns the position of \texttt{x} in \texttt{l} in $O(n)$ time.
    \item $\texttt{remove}(\texttt{l}, \texttt{x})$ removes \texttt{x} from \texttt{l} and returns its position in $O(n)$ time.
\end{itemize}

Implementing these operations reversibly in linear time required overcoming the challenges of recursive uncomputation (\Cref{sec:recursive-uncomputation}) and mutated uncomputation (\Cref{sec:mutated-uncomputation}).

\paragraph{Stack and Queue}

These operations manipulate a linked list as a LIFO stack and a FIFO queue:
\begin{itemize}
    \item $\texttt{push\_front}(\texttt{l}, \texttt{x})$ pushes \texttt{x} to the front of \texttt{l} in $O(1)$ time.
    \item $\texttt{pop\_front}(\texttt{x})$ pops an element from the front of \texttt{l} and returns it in $O(1)$ time.
    \item $\texttt{push\_back}(\texttt{l}, \texttt{x})$ pushes \texttt{x} to the back of \texttt{l} in $O(n)$ time.
\end{itemize}

The time complexity of \texttt{push\_back} is linear because we use a singly-linked list and could be further improved using a doubly-linked list.

\paragraph{String}

We implemented strings using two underlying representations: variable-length character arrays, which are not history-independent,\footnotemark{} and fixed-length words, which are. Only the word implementation is summarized in~\Cref{tbl:results}. The following operations are common to both implementations:
\footnotetext{We do not consider these arrays to be history-independent as they consist of a contiguous variable-length block rather than linked nodes that can be allocated from a pool of fixed-size blocks.}
\begin{itemize}
    \item $\texttt{is\_empty}(\texttt{s})$ returns whether \texttt{s} is an empty string in $O(1)$ time.
    \item $\texttt{length}(\texttt{s})$ returns the length of \texttt{s} in $O(1)$ time.
    \item $\texttt{get\_prefix}(\texttt{s}, \texttt{i})$ returns the prefix substring of length \texttt{i} of \texttt{s}.
    \item $\texttt{get\_substring}(\texttt{s}, \texttt{i})$ returns the suffix substring beginning at index \texttt{i} of \texttt{s}.
    \item $\texttt{get}(\texttt{s}, \texttt{i})$ returns character \texttt{i} of \texttt{s}.
    \item $\texttt{is\_prefix}(\texttt{s1}, \texttt{s2})$ returns whether \texttt{s1} is a prefix of \texttt{s2}.
    \item $\texttt{num\_matching}(\texttt{s1}, \texttt{s2})$ returns the number of characters that match in \texttt{s1} and \texttt{s2} starting from the beginning of both strings.
    \item $\texttt{equal}(\texttt{s1}, \texttt{s2})$ returns whether \texttt{s1} and \texttt{s2} are equal strings.
\end{itemize}

The time complexities of \texttt{get\_prefix}, \texttt{get\_substring}, and \texttt{get} are $O(1)$ for array-based strings, which use pointer arithmetic, and $O(k)$ for word-based strings, which use bitwise operations whose complexity increases with the word size. The complexities of \texttt{is\_prefix} and \texttt{num\_matching} are $O(n)$ for arrays and $O(\textrm{poly}(k))$ for words, which both use recursion. The complexity of \texttt{equal} is $O(n)$ for arrays, which use recursion, and $O(k)$ for words, which uses bitwise comparison.

The following operations are implemented for word-based strings only:
\begin{itemize}
    \item $\texttt{concat}(\texttt{s1}, \texttt{s2})$ returns the concatenation of \texttt{s1} and \texttt{s2} in $O(\textrm{poly}(k))$ time.
    \item $\texttt{compare}(\texttt{s1}, \texttt{s2})$ returns whether $\texttt{s1} \le \texttt{s2}$ lexicographically in $O(\textrm{poly}(k))$ time.
\end{itemize}

\paragraph{Set}

We implement the abstract data structure of sets using radix trees~\citep{morrison}. Unlike hash tables, radix trees are history-independent, and for any set of keys there exists a unique radix tree encoding the set. The set supports the interface:
\begin{itemize}
    \item $\texttt{insert}(\texttt{s}, \texttt{k}, \texttt{v})$ inserts key $\texttt{k}$ into set $\texttt{s}$ in $O(\textrm{poly}(k))$ time.
    \item $\texttt{contains}(\texttt{s}, \texttt{k})$ returns whether set $\texttt{s}$ contains key $\texttt{k}$ in $O(\textrm{poly}(k))$ time.
\end{itemize}

The radix tree structure was by far the most complex data structure to implement in the library, at over 400 lines of code in \LangName{}. Implementing an efficient tree traversal required overcoming all three challenges, particularly branch sequentialization (\Cref{sec:branch-sequentialization}).

For comparison, we also implemented sets using separate-chaining hash tables, which are not history-independent. This interface is actually more general, supporting map operations that additionally associate a key with a value:
\begin{itemize}
    \item $\texttt{insert}(\texttt{m}, \texttt{k}, \texttt{v})$ inserts $\texttt{v}$ at $\texttt{k}$ into map $\texttt{m}$ and returns the old value if present, in $O(n)$ time.
    \item $\texttt{contains}(\texttt{m}, \texttt{k})$ returns whether map $\texttt{m}$ contains key $\texttt{k}$ in $O(n)$ time.
\end{itemize}

Classically, the complexity of hash table operations is $O(n)$ in the worst case of maximal hash collisions and has $O(1)$ average complexity. The quantum implementation always incurs the $O(n)$ worst case complexity in the general case with no guarantees on hash function performance.

One may invoke \texttt{insert} twice to copy a value out of the map for a given key.
Because \texttt{insert} is reversible, it is also possible to remove the elements of a hash table in the order they were inserted by reversing the \texttt{insert} operation:
\begin{lstlisting}[style=nonum]
let null <- insert[N](m, 1, 2);
let null <- insert[N](m, 2, 4);
let null -> insert[N](m, 2, 4); /* Must remove in this order */
let null -> insert[N](m, 1, 2);
\end{lstlisting}
However, removing elements in a different order may corrupt the state of the hash table, because a separate-chaining hash table is not history-independent in general.

\end{document}